\documentclass[preprint,12pt,authoryear]{elsarticle}

\usepackage{graphicx}%
\usepackage{multirow}%
\usepackage{amsmath,amssymb,amsfonts}%
\usepackage{amsthm}%
\usepackage{mathrsfs}%
\usepackage[title]{appendix}%
\usepackage{xcolor}%
\usepackage{textcomp}%
\usepackage{manyfoot}%
\usepackage{booktabs}%
\usepackage{algorithm}%
\usepackage{algorithmicx}%
\usepackage{algpseudocode}%
\usepackage{listings}%
\usepackage{enumitem}
\usepackage{natbib}
\usepackage{bm}
\usepackage{tikz}
\usepackage{soul}
\usepackage{setspace}
\usepackage{hyperref}
\usetikzlibrary{arrows.meta, positioning, shapes.geometric}

\DeclareMathOperator{\expit}{expit}
\usepackage{changes}
\usepackage{makecell}

\theoremstyle{thmstyleone}%
\newtheorem{theorem}{Theorem}
\newtheorem{lemma}[theorem]{Lemma}         

\theoremstyle{thmstyletwo}%
\newtheorem*{remark*}{Remark}

\theoremstyle{thmstylethree}%
\newtheorem{assumption}{Assumption}  

\newcommand{\ind}{\mathbf{1}}
\newcommand{\R}{\mathbb{R}}
\newcommand{\E}{\mathbb{E}}
\newcommand{\Var}{\mathrm{Var}}

\newcommand{\dto}{\overset{d}{\longrightarrow}}
\newcommand{\pto}{\overset{p}{\longrightarrow}}

\newcommand{\Op}{\mathrm{O}_p}

\usepackage{amssymb}
\usepackage{amsmath}

\journal{}

\begin{document}

\begin{frontmatter}



\title{Beyond Point Estimates: Reliable Evaluation of Prediction Performance Metrics under Clustered Data} 


\author[1]{Taekwon Hong} 
\author[2]{Daeyoung Lim} 
\author[3]{Woojung Bae} 

\affiliation[1]{organization={Division of Biometrics VII, Office of Biostatistics, Office of Translational Sciences, Center for Drug Evaluation and Research, U.S. Food and Drug Administration},
            addressline={10903 New Hampshire Ave}, 
            city={Silver Spring},
            postcode={20993}, 
            state={Maryland},
            country={USA}}
            
\affiliation[2]{organization={Division of Biometrics III, Office of Biostatistics, Office of Translational Sciences, Center for Drug Evaluation and Research, U.S. Food and Drug Administration},
            addressline={10903 New Hampshire Ave}, 
            city={Silver Spring},
            postcode={20993}, 
            state={Maryland},
            country={USA}}

\affiliation[3]{organization={Division of Biostatistics, Office of Biostatistics and Pharmacovigilance, Center for Biologics Evaluation and Research, U.S. Food and Drug Administration},
            addressline={10903 New Hampshire Ave}, 
            city={Silver Spring},
            postcode={20993}, 
            state={Maryland},
            country={USA}}

\begin{abstract}
Prediction performance metrics such as accuracy and the F$_{1}$ score are typically reported as single numbers, with no measure of uncertainty. The omission has been tolerable in exploratory settings, where model evaluation is used for informal comparison rather than formal decision-making. But as machine learning is deployed in real-world applications, evaluation results are increasingly used to support binary decisions---whether a model meets a required standard or not---making uncertainty quantification essential. The problem is compounded when data are dependent, as in repeated measurements, clustered subjects, or time series, where variability is harder to assess and easy to underestimate. We develop a unified framework that links a broad class of performance metrics through their representation as smooth functionals of confusion-matrix probabilities. This representation allows the use of the cluster-robust sandwich variance estimator to obtain asymptotically valid confidence intervals, hypothesis tests, and paired model comparisons for both binary and multiclass problems under clustered data. We also provide power and sample size approximations based on pilot data, enabling principled study design for model evaluation. Simulations show that the proposed methods achieve near-nominal coverage across a range of dependence structures, while naive methods underestimate variability. A real-data application further illustrates how accounting for clustering can materially change conclusions. 
These results offer a practical foundation for uncertainty quantification and study design in prediction performance evaluation, in settings where decisions should be justified under dependent and clustered data.
\end{abstract}




\end{frontmatter}




\section{Introduction}
Performance evaluation is a central component of machine learning (ML) model development. The stakes are particularly high in domains such as biomedical and clinical applications, where model evaluations increasingly inform binary decision-making about whether a predictive model is suitable for deployment, driven by advances in digital health technologies (e.g., \citet{mittermaier2023digital,zwack2023evolution}). As a result, evaluation is no longer used solely for informal model comparison; instead, it should be grounded in statistically principled methods for determining whether predefined performance criteria have been met. This shift raises the bar for statistical rigor in how model performance is assessed.

In practice, model performance is typically summarized using scalar metrics derived from the confusion matrix, such as accuracy, precision, recall, the F$_{1}$ score, and the Matthews correlation coefficient \citep[MCC]{matthews1975comparison} (see, for example, \citet{chicco2020advantages}). Despite their widespread use, these metrics are often reported as point estimates without measures of uncertainty (see, for example, \citet{sathyanarayana2016sleep,esteva2017dermatologist}). This is somewhat surprising, given that uncertainty quantification is routine in most statistical analyses, yet remains uncommon in the reporting of ML performance metrics. When uncertainty is considered, it is typically under the assumption that observations are independently and identically distributed (i.i.d.).

In many, if not most, applications, however, this i.i.d. assumption does not hold. Evaluation data are often inherently dependent, with observations grouped into clusters (i.e., sets of observations arising from the same subject or experimental unit), such as repeated measurements from the same patient \citep{avey2025home}, multiple images per subject \citep{fujioka2019distinction}, or temporally correlated records \citep{abnoosian2023prediction}. Ignoring this dependence can substantially underestimate variability, leading to overly narrow confidence intervals (CIs) and inflated Type I error rates. While cluster-robust methods are well-established in statistics and econometrics (e.g., \citet{liang1986longitudinal, colin2015practitioner, abadie2023should}), their use in ML evaluation remains underexplored.

Recent work has begun to address statistical inference for specific performance metrics. In particular, large-sample CIs and hypothesis tests have been developed for measures such as the F$_{1}$ score and MCC using delta-method-based approaches \citep{takahashi2022confidence, takahashi2023hypothesis, itaya2025asymptotic}. However, these methods are generally metric-specific and formulated under independent observations, without accounting for clustered or within-subject dependence. As a result, they do not extend to more realistic evaluation settings with correlated data, nor do they provide a unified framework applicable across a broad class of performance measures.

In many biomedical applications, particularly those involving sensitive and costly-to-collect data, evaluation datasets are not easily standardized or shared, and analyses must often rely on relatively small samples. In these settings, quantifying uncertainty in performance metrics and determining appropriate sample sizes are essential for reliable conclusions, yet formal tools for doing so are lacking.

In this paper, we propose a unified framework for statistical inference on prediction performance metrics under clustered data, as commonly encountered in biomedical applications. We represent performance measures as smooth functionals of confusion-matrix cell probabilities and leverage large-sample approximations, via the delta method (see, e.g., \citet{oehlert1992note}), together with sandwich variance estimation to account for within-cluster dependence. This approach yields asymptotically valid CIs and hypothesis tests for a wide class of metrics in both binary and multiclass settings, without requiring metric-specific derivations.

Within this framework, we develop inference procedures for both single-model evaluation and paired model comparison, namely, superiority and noninferiority testing \citep{walker2011understanding}. We further extend the framework to support power analysis and sample size determination using pilot-based variance estimates, enabling prospective, study-style evaluation of model performance with prespecified design criteria.

We evaluate the proposed methods through simulation studies and illustrate their practical use with a real-data application. The results demonstrate that accounting for clustering yields well-calibrated uncertainty estimates and near-nominal coverage, whereas naive approaches that disregard dependence can substantially underestimate variability and produce misleading conclusions.

The remainder of this paper is organized as follows. Section~\ref{sec:method} presents the proposed inferential framework for single-model inference and paired comparisons, along with power and sample size approximations for study design. Section~\ref{sec:simulation} reports simulation studies assessing finite-sample performance and empirical power. Section~\ref{sec:application} showcases the use of the proposed methods in a real-data example. Section~\ref{sec:discussion} concludes with limitations and future research directions.

\section{Methodology}\label{sec:method}

\subsection{Notation and Settings}
We consider classification performance evaluation under clustered data. Our goal is to develop an inferential framework for performance metrics that can be expressed as smooth functions of confusion-matrix cell probabilities.

Let clusters be indexed by $i = 1, \ldots, n$, with cluster $i$ contributing $m_{ni} \geq 1$ observations, and total sample size $N_n = \sum_{i=1}^n m_{ni}$. For observation $j$ in cluster $i$, let $Y_{nij}$ and $\widehat{Y}_{nij}$ denote the true and predicted labels. Define the indicator vector $\bm{Z}_{nij} \in \mathbb{R}^{d}$, where $d = r^2$ in the $r$-class setting, with entries
$$
Z_{nij}^{(ab)} = \ind(\widehat{Y}_{nij}=a, Y_{nij}=b), \quad (a,b) \in \{1,\ldots,r\}^2.
$$
The empirical cell-probability estimator is
$$
\widehat{\bm{p}}_n = \frac{1}{N_n}\sum_{i=1}^n \sum_{j=1}^{m_{ni}} \bm{Z}_{nij} = \frac{1}{N_{n}}\sum_{i=1}^n \bm{S}_{ni},
$$
where $\bm{S}_{ni}=\sum_{j=1}^{m_{ni}} \bm{Z}_{nij}$ denotes the cluster-level contribution. Let $\bm{p}_n = \E(\widehat{\bm{p}}_n)$ denote the corresponding population cell-probability vector. Then,
$$
\widehat{\bm p}_n - \bm p_n
=
\frac{1}{N_n}\sum_{i=1}^n \bm U_{ni}, \quad \bm{U}_{ni} = \bm{S}_{ni} - m_{ni}\bm{p}_n,
$$
which expresses the estimator as a sum of independent, mean-zero cluster-level terms.

We consider performance metrics of the form:
$$
\theta_n = g(\bm{p}_n),\quad \widehat{\theta}_n = g(\widehat{\bm{p}}_n),
$$
where $g: \mathcal{D} \subset \mathbb{R}^d \to \mathbb{R}$ is a smooth function, for which $\mathcal{D}$ is an open set containing the relevant parameter values. This formulation covers a wide range of commonly used metrics. For example, in the binary case ($r=2$) with $\bm{p}=(p_{11}, p_{10}, p_{01}, p_{00})^\top$, 
\begin{align*}
    &\text{Accuracy}=p_{11}+p_{00}, \quad \text{Sensitivity}=\frac{p_{11}}{p_{11}+p_{01}}, \quad \text{Specificity}=\frac{p_{00}}{p_{10}+p_{00}},\\
    &F_1 =\frac{2p_{11}}{2p_{11}+p_{10}+p_{01}},\; \text{MCC}=\frac{p_{11}p_{00}-p_{10}p_{01}}
{\sqrt{(p_{11}+p_{10})(p_{11}+p_{01})(p_{00}+p_{10})(p_{00}+p_{01})}},
\end{align*}
all of which are smooth on appropriate interior domains. In the multiclass setting, metrics such as micro- and macro-F$_{1}$ admit analogous representations \citep{manning2008introduction} .

The mathematical underpinning of our framework is that, if $g$ is sufficiently smooth and $\widehat{\bm{p}}_n$ is asymptotically normal under clustering, then inference for $\theta_n$ follows from the multivariate delta method.

\subsection{Inferential framework}

We formalize model evaluation as a problem of statistical inference on performance metrics. In this framework, evaluation questions become formal hypothesis tests: for example, whether a model exceeds a target level of performance or differs meaningfully from a comparator. We treat the trained model as fixed and perform inference on an out-of-sample test dataset to validate performance and assess whether observed results generalize beyond chance \citep{chekroud2024illusory}.

Within this setup, we focus on two complementary types of inference that commonly arise in practice: superiority testing and noninferiority testing. \emph{Superiority testing} asks whether a model's performance exceeds a fixed reference level $\theta_0$. This is useful, for instance, when determining whether a model meets a minimum deployment threshold or outperforms a baseline system. \emph{Noninferiority testing} asks whether a candidate model is not worse than a reference model by more than a prespecified margin $\Delta$ (also known as noninferiority margin). This allows for statistically justified model replacement even when strict superiority cannot be established or is not necessary. It is especially relevant in paired comparisons, where both models are evaluated on the same dataset and performance metrics; for example, when comparing confusion-matrix probabilities from a candidate model $p_{n}^{(c)}$ and a reference model $p_{n}^{(f)}$ under a shared performance metric.

Both types of testing can be handled within our framework by combining the cluster-level central limit theory (CLT) with delta-method variance propagation. We now introduce the required regularity conditions. Throughout, $\pto$ and $\dto$ denote convergence in probability and in distribution, respectively \citep{billingsley1999convergence}, and deterministic limits are denoted by $\to$.

\begin{assumption}[Independent clusters]\label{ass:mc-ind}
For each $n$, the cluster-level arrays
\[
\{(Y_{nij},\widehat Y_{nij}): j=1,\dots,m_{ni}\},
\qquad i=1,\dots,n,
\]
are independent across $i$. Within-cluster dependence is unrestricted. 
\end{assumption}

\begin{assumption}[Uniformly bounded cluster sizes]\label{ass:mc-bound}
There exists a finite constant $M<\infty$ such that
\[
1\le m_{ni}\le M
\qquad \text{for all } n,i,
\]
and
\[
\frac{N_n}{n}\to \bar m
\qquad \text{for some } \bar m\in(0,M].
\]
\end{assumption}

\begin{assumption}[Row-wise mean homogeneity]\label{ass:mc-homog}
For each $n$, there exists a vector $\bm p_n \in (0,1)^{r^2}$ such that
\[
\E(\bm Z_{nij}) = \bm p_n
\qquad \text{for all } i=1,\dots,n,\ j=1,\dots,m_{ni}.
\] Then, it follows that
\[
\E(\bm S_{ni}) = m_{ni}\bm p_n.
\]
\end{assumption}

\begin{assumption}[Stable multiclass cluster covariance]\label{ass:mc-cov}
There exists a finite positive semidefinite matrix $\bm \Omega_r$ of dimension $r^2\times r^2$ such that
\[
\bm \Omega_{r,n}
:=
\frac{1}{N_n}\sum_{i=1}^n \Var(\bm U_{ni})
\to
\bm \Omega_r.
\]
\end{assumption}

Assumption \ref{ass:mc-ind} posits independence across clusters while allowing arbitrary dependence within each cluster, which aligns with common settings such as repeated measurements or multiple samples per individual. Assumption \ref{ass:mc-bound} ensures that no single cluster dominates the sample and that the average cluster size remains well-behaved as $n$ increases. This is a condition typically satisfied when each cluster contributes only a limited number of observations. Assumption~\ref{ass:mc-homog} states that all observations within a given sample share a common confusion-matrix distribution, so that cluster-level differences arise from dependence and cluster size rather than systematic shifts in class probabilities. Finally, Assumption \ref{ass:mc-cov} requires that the normalized aggregate variability from independent clusters converges to a finite limit. Together, Assumptions~\ref{ass:mc-ind}--\ref{ass:mc-cov} describe a regime with independent clusters, bounded cluster sizes, a common underlying confusion-matrix distribution, and well-behaved aggregate variability.

\begin{theorem}[Consistency of the sandwich covariance estimator]\label{thm:sandwich}
Consider the cluster-robust covariance estimator
\[
\widehat{\bm \Omega}_{r,n}
=
\frac{1}{N_n}\sum_{i=1}^n
\widehat{\bm U}_{ni}\widehat{\bm U}_{ni}^\top.
\]
Under Assumptions \ref{ass:mc-ind}, \ref{ass:mc-bound}, \ref{ass:mc-homog}, and \ref{ass:mc-cov},
\[
\widehat{\bm \Omega}_{r,n} \pto \bm \Omega_r.
\]
\end{theorem}

Theorem \ref{thm:sandwich} ensures that the variability induced by clustered observations can be consistently estimated directly from the data, without requiring explicit modeling of within-cluster dependence, providing a robust and scalable approach to uncertainty quantification in performance metrics under complex or unknown correlation structures. 

\begin{assumption}[Stable multiclass cell-probability vector]\label{ass:mc-first}
There exists
\[
\bm p
=
(p_{11},p_{12},\dots,p_{rr})^\top
\in (0,1)^{r^2},
\qquad
\bm 1^\top \bm p=1,
\]
such that
\[
\bm p_n \to \bm p.
\]
\end{assumption}

\begin{assumption}[Smooth metric and interiority]\label{ass:smooth}
The target metric is of the form $\theta_n=g(\bm p_n)$ for a function $g$ that is continuously differentiable on an open neighborhood of $\bm p$. In addition, all denominators appearing in $g$ are bounded away from zero in a neighborhood of $\bm p$.
\end{assumption}

Assumption \ref{ass:mc-first} ensures that the confusion-matrix probabilities converge to a well-defined limit, providing a stable target parameter. Assumption \ref{ass:smooth} imposes smoothness of the target metric with respect to the underlying probabilities, excluding degenerate boundary cases (e.g., metrics evaluated at extreme values such as 0\% or 100\% precision); the common example metrics satisfy this condition. All Assumptions \ref{ass:mc-ind}--\ref{ass:smooth} are analogous in spirit to the regularity conditions in \citet{liang1986longitudinal}. The difference is that rather than specifying a parametric mean model together with a working correlation structure, we work directly with the confusion-matrix probabilities and do not model the mean.

\subsubsection{Superiority testing}

Suppose a researcher wants to evaluate whether a newly developed model achieves a meaningful level of performance on a chosen metric. A natural point of reference is a prespecified benchmark $\theta_0$, representing a target aligned with practical expectations in the application. In practice, such benchmarks are often informed by domain standards and evaluated on data that are representative of the population of interest. When the benchmark and the evaluation data are not aligned in this way, performance comparisons can be hard to interpret and may not provide a sound basis for downstream decision-making.

To support principled decision making, the evaluation problem can be formalized as a hypothesis test. For metrics where larger values indicate better performance, superiority testing asks whether a model exceeds the target level:
\[
H_0:\theta \le \theta_0
\qquad \text{versus} \qquad
H_1:\theta > \theta_0,
\]
where $H_0$ is the null hypothesis and $H_1$ is the alternative of interest. For metrics where smaller values are preferable, the direction of the inequality is reversed.

In practice, the true performance $\theta$ is unknown and must be estimated from test data; we write this as $\widehat\theta_n = g(\widehat{\bm p}_n)$. The following theorem characterizes the large-sample behavior of $\widehat\theta_n$ under clustered observations, which forms the basis for inference.

\begin{theorem}[Delta method for a single-model inference]\label{thm:delta-general}
Let $\nabla$ denote the gradient. Under Assumptions~\ref{ass:mc-ind}--\ref{ass:mc-first},
\[
\sqrt{N_n}\left(\widehat{\bm p}_n - \bm p_n\right)
\dto
\mathcal{N}(\bm 0, \bm \Omega_r).
\]
If, in addition, Assumptions~\ref{ass:smooth} holds, then
\[
\sqrt{N_n}(\widehat\theta_n - \theta_n)
\dto
N(0, V),
\]
where
\[
V
=
\nabla g(\bm p)^\top \bm \Omega_r \nabla g(\bm p).
\]
\end{theorem}
For inference on $\theta_n$, consider the plug-in variance estimator
\[
\widehat V_n
=
\nabla g(\widehat{\bm p}_n)^\top
\widehat{\bm \Omega}_{r,n}
\nabla g(\widehat{\bm p}_n),
\]
which satisfies $\widehat V_n \pto V$. Theorem \ref{thm:delta-general} shows that, after appropriate normalization, the nonlinear metric $\widehat{\theta}_n = g(\widehat{\bm{p}}_n)$ behaves like a linear function of $\widehat{\bm{p}}_n$ and inherits its asymptotic Gaussian distribution, even under within-cluster dependence.

Theorem \ref{thm:delta-general} provides the basis for statistical testing. Assume that $\hat{V}_n,V> 0
$, i.e., nondegenerate variance. If the null hypothesis is true at the boundary $\theta=\theta_0$, then
\[
T_n^{\mathrm{sup}}=\frac{\sqrt{N_n}(\widehat\theta_n-\theta_0)}{\sqrt{\widehat V_n}} \dto N(0,1).
\]
A two-sided $100(1-\alpha)\%$ asymptotic CI for $\theta_n$ is given by
\[
\bigg(\widehat\theta_n - z_{1-\alpha/2}\sqrt{\widehat V_n/N_n},\; \widehat\theta_n+z_{1-\alpha/2}\sqrt{\widehat V_n/N_n}\bigg),
\]
while the one-sided $100(1-\alpha)\%$ CI is given by
\[
\bigg[\widehat\theta_n - z_{1-\alpha}\sqrt{\widehat V_n/N_n},\;\infty \bigg).
\]
The boundary case $\theta=\theta_0$ represents the least favorable configuration under the null hypothesis, ensuring uniform control of the Type I error over the null parameter space $\theta \le \theta_0$. Therefore, we reject the null hypothesis $H_0$ if $T_n^{\mathrm{sup}} > z_{1-\alpha}$, corresponding to an $\alpha$-level test. Such a rejection indicates that the observed data are inconsistent with $H_0$ at significance level $\alpha$. Equivalently, by test-CI duality, we reject $H_0$ if the one-sided $100(1-\alpha)\%$ CI for $\theta$ excludes $\theta_0$. 

\subsubsection{Noninferiority testing}
In practice, the goal may not necessarily be to demonstrate that a new model strictly outperforms an existing one, but rather to establish that it is not meaningfully worse. This situation often arises when comparing a candidate model to a reference model on the same dataset, especially when the candidate offers practical advantages (such as simpler implementation, better interpretability, or lower cost) that are not captured by predictive performance.

Let $\bm p_n^{(c)}$ and $\bm p_n^{(f)}$ denote the population confusion vectors for the candidate and reference models, respectively, and let $\widehat{\bm p}_n^{(c)}$ and $\widehat{\bm p}_n^{(f)}$ denote their estimators. Under Assumption~\ref{ass:mc-first}, these estimators converge to limiting vectors $\bm p^{(c)}$ and $\bm p^{(f)}$, respectively. Smooth performance metrics for the two models and their estimators are defined analogously.

Let $d = \theta^{(c)} - \theta^{(f)}$ denote the difference in population performance. Given a prespecified noninferiority margin $\Delta>0$, noninferiority testing asks whether any loss in performance of the candidate model, relative to the reference, falls within an acceptable tolerance:
\[
H_0: d \le -\Delta
\qquad \text{versus} \qquad
H_1: d > -\Delta.
\]
Here, $\Delta$ should be chosen to reflect a substantively meaningful threshold, typically informed by domain knowledge, so that differences smaller than $\Delta$ are not regarded as practically important. Under this formulation, the null hypothesis $H_0$ corresponds to the candidate model being unacceptably worse than the reference, whereas $H_1$ captures the case where the candidate is not inferior beyond the margin $\Delta$.

Because both models are evaluated on the same clustered dataset, their estimators are statistically dependent. Accounting for this dependence requires characterizing their joint large-sample behavior. For each cluster \(i=1,\dots,n\), let
$\bm S_{ni}^{(c)}$ and $\bm S_{ni}^{(f)}$
denote the cluster-level confusion-count vectors for the candidate and reference models. By Assumption~\ref{ass:mc-homog},
\[
\E(\bm S_{ni}^{(c)}) = m_{ni}\bm p_n^{(c)},
\qquad
\E(\bm S_{ni}^{(f)}) = m_{ni}\bm p_n^{(f)}.
\]
Then, define the centered cluster contributions as
$$
\bar{\bm U}_{ni}
=
\begin{pmatrix}
\bm U_{ni}^{(c)}\\
\bm U_{ni}^{(f)}
\end{pmatrix}=\begin{pmatrix}
\bm S_{ni}^{(c)}-m_{ni}\bm p_n^{(c)}\\
\bm S_{ni}^{(f)}-m_{ni}\bm p_n^{(f)}
\end{pmatrix}
\in \R^{2r^2}.
$$

To derive the joint asymptotic variance, we impose the following additional block covariance assumption.

\begin{assumption}[Stable joint cluster covariance]\label{ass:mc-cov-pair}
There exists a finite positive semidefinite matrix
\[
\bar{\bm \Omega}_r
=
\begin{pmatrix}
\bm \Omega_r^{(c)} & \bm \Omega_r^{(cf)}\\
\bm \Omega_r^{(fc)} & \bm \Omega_r^{(f)}
\end{pmatrix}
\]
such that
\[
\bar{\bm \Omega}_{r,n}
:=
\frac{1}{N_n}\sum_{i=1}^n \Var(\bar{\bm U}_{ni})
\to
\bar{\bm \Omega}_r.
\]
\end{assumption}

The following result gives the joint asymptotic distribution of the empirical confusion vectors and the induced asymptotic distribution of the performance difference.

\begin{theorem}[Delta method for a paired comparison of two models]\label{thm:pair}
Under Assumptions~\ref{ass:mc-ind}--\ref{ass:mc-homog}, \ref{ass:mc-first}, and \ref{ass:mc-cov-pair},
\[
\sqrt{N_n}
\begin{pmatrix}
\widehat{\bm p}_n^{(c)} - \bm p_n^{(c)} \\
\widehat{\bm p}_n^{(f)} - \bm p_n^{(f)}
\end{pmatrix}
\dto
\mathcal{N}\left(
\bm 0,
\begin{pmatrix}
\bm \Omega_r^{(c)} & \bm \Omega_r^{(cf)} \\
\bm \Omega_r^{(fc)} & \bm \Omega_r^{(f)}
\end{pmatrix}
\right).
\]
If, in addition, Assumption~\ref{ass:smooth} holds, then
\[
\sqrt{N_n}(\widehat d_n - d_n)
\dto
N(0, V_d),
\]
where
\[
V_d =
\nabla g(\bm p^{(c)})^\top \bm \Omega_r^{(c)} \nabla g(\bm p^{(c)})
+
\nabla g(\bm p^{(f)})^\top \bm \Omega_r^{(f)} \nabla g(\bm p^{(f)})
- 2\, \nabla g(\bm p^{(c)})^\top \bm \Omega_r^{(cf)} \nabla g(\bm p^{(f)}).
\]
\end{theorem}

See the \ref{appendix:proof} for the proofs of Theorems~\ref{thm:sandwich}--\ref{thm:pair}. For inference on \(d_n\), consider the plug-in variance estimator
\[
\widehat V_{d,n}
=
\nabla g(\widehat{\bm p}_n^{(c)})^\top
\widehat{\bm \Omega}_{r,n}^{(c)}
\nabla g(\widehat{\bm p}_n^{(c)})
+
\nabla g(\widehat{\bm p}_n^{(f)})^\top
\widehat{\bm \Omega}_{r,n}^{(f)}
\nabla g(\widehat{\bm p}_n^{(f)})
- 2\, \nabla g(\widehat{\bm p}_n^{(c)})^\top
\widehat{\bm \Omega}_{r,n}^{(cf)}
\nabla g(\widehat{\bm p}_n^{(f)}),
\]
which satisfies \(\widehat V_{d,n} \pto V_d\). The cross-covariance term \(\widehat{\bm \Omega}_{r,n}^{(cf)}\) accounts for the dependence induced by evaluating both models on the same clustered units.

Assume that \(\widehat V_{d,n},V_d>0\). Under the boundary case \(d = -\Delta\) of the null hypothesis,
\[
T_n^{\mathrm{NI}}
=
\frac{\sqrt{N_n}(\widehat d_n+\Delta)}{\sqrt{\widehat V_{d,n}}}
\dto N(0,1).
\]
Again, the boundary case represents the least favorable configuration under $H_0$, ensuring validity for all $d \le -\Delta$. We therefore reject $H_0$ at level $\alpha$ if $T_n^{\mathrm{NI}} > z_{1-\alpha}$, corresponding to an $\alpha$-level noninferiority test. Such a rejection indicates that the observed data are inconsistent with inferiority beyond the prespecified margin $\Delta$. Again, by the test--confidence interval duality, we reject $H_0$ if the lower confidence bound for $d$, namely
\[
\bigg[\widehat d_n - z_{1-\alpha}\sqrt{\widehat V_{d,n}/N_n},\infty\bigg)
\]
exceeds \(-\Delta\). 

\subsubsection{Asymptotic power and sample size approximations}\label{sec:method_samplesize}

In the study-design phase, before formal hypothesis testing, it is important to determine how much data is needed to reliably detect a practically meaningful difference in performance. If the sample size is too small, real improvements may go undetected, resulting in low statistical power and a higher risk of Type II error. On the other hand, an overly large sample size can create unnecessary computational, logistical, or even ethical burden, particularly in biomedical applications involving clustered or correlated data. In this section, we develop power calculations and sample size formulas for both superiority and noninferiority testing using large-sample approximations based on the delta method.

We begin with \emph{superiority testing}. Under a fixed alternative $\theta=\theta_1>\theta_0$ and $V>0$, Theorem \ref{thm:delta-general} implies that
\[
T_n^{\mathrm{sup}}
\approx
N\!\left(
\frac{\sqrt{N_n}(\theta_1-\theta_0)}{\sqrt{V}},
\,1
\right).
\]
An analogous formulation applies when $\theta_1<\theta_0$ under the reversed direction of the hypothesis (i.e., when smaller values indicate better performance). Here, $\theta_0$ defines the benchmark or decision threshold, while $\theta_1$ represents the expected true performance of the model in the validation setting. Although any $\theta_1>\theta_0$ is mathematically admissible, in practice $\theta_1$ should reflect a realistic and practically relevant performance level informed by prior knowledge, domain expertise, or empirical evidence (e.g., from training procedures for model construction). The choice drives the required sample size for achieving adequate power. The asymptotic power for superiority testing $\pi_{\sup}$ can then be approximated as
\[
\pi_{\mathrm{sup}}(N_n)
\approx
\Phi\!\left(
\frac{\sqrt{N_n}(\theta_1-\theta_0)}{\sqrt{V}}
-
z_{1-\alpha}
\right).
\]

Solving for the total number of observations $N_n$ needed to achieve a target power $1-\beta$ yields the first-order approximation
\[
N_n
\approx
\frac{(z_{1-\alpha}+z_{1-\beta})^2 V}{(\theta_1-\theta_0)^2}.
\]
If the expected average cluster size is $\bar m$, the corresponding number of clusters is approximately $n \approx N_n/\bar m$. This expression makes the trade-off clear: smaller performance gains (i.e., $\theta_1-\theta_0$) require much larger samples to be detected reliably, especially with clustered data, where additional sources of variability contribute to $V$.

We next turn to \emph{noninferiority testing} under a fixed alternative $d=d_1>-\Delta$, where $d_1$ represents the anticipated performance difference between the candidate and reference models. Under Theorem~\ref{thm:pair} and the nondegeneracy assumption $V_d>0$,
\[
T_n^{\mathrm{NI}}
\approx
N\!\left(
\frac{\sqrt{N_n}(d_1+\Delta)}{\sqrt{V_d}},
\,1
\right).
\]
$d_1$ should be specified analogously to $\theta_1$ in the superiority setting. The corresponding asymptotic power is approximated as
\[
\pi_{\mathrm{NI}}(N_n)
\approx
\Phi\!\left(
\frac{\sqrt{N_n}(d_1+\Delta)}{\sqrt{V_d}}
-
z_{1-\alpha}
\right),
\]
and solving for $N_n$ to achieve target power $1-\beta$ yields
\[
N_n
\approx
\frac{(z_{1-\alpha}+z_{1-\beta})^2 V_d}{(d_1+\Delta)^2},
\]
with the number of clusters again given by $n \approx N_n/\bar m$. 
In this setting, both the anticipated difference $d_1$ and the margin $\Delta$ play a critical role: tighter margins and smaller expected differences require substantially larger sample sizes to establish noninferiority with high confidence.

The remaining unknowns in these formulas are the asymptotic variances $V$ and $V_d$, which generally do not have simple closed-form expressions under clustered observations. In practice, a common approach is to leverage pilot data, such as training or development datasets used during model construction, to obtain plug-in estimates of $V$ and $V_d$, provided that these data are reasonably representative of the target population. 

\section{Simulation study}\label{sec:simulation}

\subsection{Objectives}
We conduct a simulation study with two complementary objectives. The primary objective is to examine the finite-sample performance of the proposed cluster-robust inference procedure under data-generating mechanisms that satisfy the conditions required for the large-sample approximations. Specifically, we assess (i) the bias of the plug-in estimator for the target metric, (ii) the agreement between the empirical standard error (ESE) and the corresponding asymptotic standard error (ASE), and (iii) the coverage probability (CP) of the resulting CIs.

The secondary objective is to evaluate the practical accuracy of the proposed asymptotic power and sample size approximations. Specifically, we examine whether the first-order formulas in Section~\ref{sec:method_samplesize} reliably yield sample sizes that achieve a desired level of power.

\subsection{Finite-sample performance of asymptotic inference}

\subsubsection{Data-generating mechanism}

Clustered observations are generated at the cluster level with $2000$ Monte Carlo replicates. Let $n\in\{50,100,200\}$ denote the number of clusters. Each cluster contributes $m_i$ lower-level observations, where $m_i$ is sampled independently and uniformly from $\{100,\dots,300\}$. This setup produces moderately large cluster sizes, consistent with the real-data analysis, discussed in Section~\ref{sec:application}, and satisfies the bounded cluster size assumption while still allowing for substantial within-cluster dependence.

Within-cluster dependence is induced using a latent Gaussian copula (see, e.g., \citet{joe2014dependence}). For each cluster, we generate a latent multivariate normal vector with either a compound symmetry (CS) or autoregressive (AR) model of order 1 correlation structure, with correlation parameter $\rho\in\{0.5,0.8\}$. These latent variables are transformed to uniform variables and mapped to multinomial confusion-matrix cells via inverse-CDF thresholding. This approach preserves the target marginal cell probabilities while introducing controlled within-cluster dependence. Clusters are generated independently.

The CS and AR(1) structures are used only to generate controlled within-cluster dependence in the simulation study; the proposed inferential procedure remains agnostic to the dependence structure and applicable under arbitrary within-cluster dependence, including nested dependence, as long as clusters are defined at the highest independent level.

\subsubsection{Target metrics}

For the binary classification setting, we consider two scenarios: balanced and imbalanced. The balanced scenario assumes prevalence of $0.5$ for the positive class with sensitivity $0.7$ and specificity $0.7$. The imbalanced scenario assumes the prevalence $0.2$ with sensitivity $0.8$ and specificity $0.9$. 

We evaluate three metrics:
\begin{itemize}[nosep]
    \item Sensitivity, representing performance on the positive class,
    \item Specificity, representing performance on the negative class, and
    \item MCC, a nonlinear summary that incorporates all four entries of the confusion matrix.
\end{itemize}
Taken together, these metrics let us compare behavior across both simple ratio-based measures and a more complex nonlinear functional. For the multiclass setting, we consider a three-class classification problem: both balanced and imbalanced based on prevalence of class. We evaluate:
\begin{itemize}[nosep]
    \item micro-F$_{1}$, which reduces to overall accuracy in the single-label setting, and
    \item macro-F$_{1}$, which averages performance across classes and is sensitive to class imbalance.
\end{itemize}

\subsubsection{Results}
\begin{table}[tbp]\footnotesize
\centering
\caption{Simulation results in the binary setting under stronger within-cluster dependence (\(\rho=0.8\)). ESE indicates empirical standard error. ASE$_{\mathrm{rob}}$ (CP$_{\mathrm{rob}}$) and ASE$_{\mathrm{naive}}$ (CP$_{\mathrm{naive}}$) denote cluster-robust and naive asymptotic standard errors with corresponding coverage probabilities. True performance values, bias, and all SEs are reported on the \(\times 100\) scale, and CP is reported as a percentage (ideally 95\%).}
\label{tab:sim_binary_rho08}
\begin{tabular}{lccrrrrr}
\toprule
Metric & $n$ & Corr. & True & Bias & ESE & ASE$_{\mathrm{rob}}$ (CP) & ASE$_{\mathrm{naive}}$ (CP) \\
\midrule
\multicolumn{8}{l}{\textbf{Balanced:} $P(\text{Class 1})=0.5$, $P(\text{Class 2})=0.5$}\\
\midrule
Sensitivity & 50  & AR1 & 70.0 & 0.0  & 1.1 & 1.1 (94.0) & 0.6 (74.2) \\
Sensitivity & 50  & CS  & 70.0 & -0.3 & 5.0 & 5.0 (94.2) & 0.6 (19.9) \\
Sensitivity & 100 & AR1 & 70.0 & 0.0  & 0.8 & 0.8 (95.3) & 0.5 (75.8) \\
Sensitivity & 100 & CS  & 70.0 & -0.4 & 3.7 & 3.6 (94.3) & 0.5 (19.6) \\
Sensitivity & 200 & AR1 & 70.0 & 0.0  & 0.6 & 0.5 (94.9) & 0.3 (74.3) \\
Sensitivity & 200 & CS  & 70.0 & 0.0  & 2.6 & 2.5 (94.3) & 0.3 (20.3) \\

Specificity & 50  & AR1 & 70.0 & 0.0  & 1.1 & 1.1 (93.6) & 0.6 (74.5) \\
Specificity & 50  & CS  & 70.0 & -0.3 & 5.2 & 5.0 (93.6) & 0.6 (18.4) \\
Specificity & 100 & AR1 & 70.0 & 0.0  & 0.8 & 0.8 (94.0) & 0.5 (74.4) \\
Specificity & 100 & CS  & 70.0 & -0.1 & 3.7 & 3.6 (94.2) & 0.5 (19.7) \\
Specificity & 200 & AR1 & 70.0 & 0.0  & 0.5 & 0.5 (94.7) & 0.3 (76.1) \\
Specificity & 200 & CS  & 70.0 & -0.1 & 2.6 & 2.5 (95.3) & 0.3 (19.2) \\

MCC & 50  & AR1 & 40.0 & 0.0  & 1.3 & 1.3 (93.2) & 0.9 (82.6) \\
MCC & 50  & CS  & 40.0 & -0.5 & 6.6 & 6.5 (94.0) & 0.9 (18.8) \\
MCC & 100 & AR1 & 40.0 & 0.0  & 0.9 & 0.9 (94.8) & 0.6 (83.8) \\
MCC & 100 & CS  & 40.0 & -0.4 & 4.7 & 4.6 (94.4) & 0.7 (19.9) \\
MCC & 200 & AR1 & 40.0 & 0.0  & 0.7 & 0.6 (94.8) & 0.5 (83.0) \\
MCC & 200 & CS  & 40.0 & 0.0  & 3.2 & 3.3 (94.5) & 0.5 (24.1) \\

\midrule
\multicolumn{8}{l}{\textbf{Imbalanced:} $P(\text{Class 1})=0.8$, $P(\text{Class 2})=0.2$}\\
\midrule
Sensitivity & 50  & AR1 & 80.0 & 0.0  & 1.1 & 1.1 (93.1) & 0.9 (87.8) \\
Sensitivity & 50  & CS  & 80.0 & -0.7 & 4.6 & 4.3 (92.3) & 0.9 (32.9) \\
Sensitivity & 100 & AR1 & 80.0 & 0.0  & 0.8 & 0.8 (95.2) & 0.6 (88.8) \\
Sensitivity & 100 & CS  & 80.0 & -0.3 & 3.1 & 3.0 (94.1) & 0.6 (31.8) \\
Sensitivity & 200 & AR1 & 80.0 & 0.0  & 0.6 & 0.6 (95.0) & 0.4 (88.1) \\
Sensitivity & 200 & CS  & 80.0 & -0.1 & 2.1 & 2.2 (94.4) & 0.4 (32.0) \\

Specificity & 50  & AR1 & 90.0 & 0.0  & 0.5 & 0.5 (94.2) & 0.3 (80.0) \\
Specificity & 50  & CS  & 90.0 & 0.0  & 2.0 & 2.0 (93.3) & 0.3 (26.2) \\
Specificity & 100 & AR1 & 90.0 & 0.0  & 0.4 & 0.3 (94.0) & 0.2 (81.7) \\
Specificity & 100 & CS  & 90.0 & 0.0  & 1.4 & 1.4 (94.5) & 0.2 (25.0) \\
Specificity & 200 & AR1 & 90.0 & 0.0  & 0.2 & 0.2 (94.4) & 0.2 (80.7) \\
Specificity & 200 & CS  & 90.0 & 0.0  & 1.0 & 1.0 (94.3) & 0.2 (26.5) \\

MCC & 50  & AR1 & 65.6 & -0.1 & 1.2 & 1.1 (93.2) & 0.9 (86.1) \\
MCC & 50  & CS  & 65.6 & -0.7 & 5.1 & 4.8 (91.9) & 0.9 (28.7) \\
MCC & 100 & AR1 & 65.6 & 0.0  & 0.8 & 0.8 (94.4) & 0.6 (87.9) \\
MCC & 100 & CS  & 65.6 & -0.3 & 3.5 & 3.5 (93.6) & 0.6 (27.8) \\
MCC & 200 & AR1 & 65.6 & 0.0  & 0.6 & 0.6 (95.5) & 0.5 (89.3) \\
MCC & 200 & CS  & 65.6 & -0.1 & 2.5 & 2.5 (94.3) & 0.5 (28.6) \\
\bottomrule
\end{tabular}
\end{table}

\begin{table}[tbp]\footnotesize
\centering
\caption{Simulation results in the multiclass setting under stronger within-cluster dependence (\(\rho=0.8\)). ESE indicates empirical standard error. ASE$_{\mathrm{rob}}$ (CP$_{\mathrm{rob}}$) and ASE$_{\mathrm{naive}}$ (CP$_{\mathrm{naive}}$) denote cluster-robust and naive asymptotic standard errors with corresponding coverage probabilities. True performance values, bias, and all SEs are reported on the \(\times 100\) scale, and CP is reported as a percentage (ideally 95\%).}
\label{tab:sim_multiclass_rho08}
\begin{tabular}{lccrrrrr}
\toprule
Metric & $n$ & Corr. & True & Bias & ESE & ASE$_{\mathrm{rob}}$ (CP) & ASE$_{\mathrm{naive}}$ (CP) \\
\midrule
\multicolumn{8}{l}{\textbf{Balanced:} $P(\text{Class 1})=0.3$, $P(\text{Class 2})=0.37$, $P(\text{Class 3})=0.33$}\\
\midrule
macro-F$_{1}$ & 50  & AR1 & 82.0 & 0.0  & 0.4 & 0.4 (93.3) & 0.4 (91.6) \\
macro-F$_{1}$ & 50  & CS  & 82.0 & -0.3 & 1.4 & 1.4 (93.3) & 0.4 (41.6) \\
macro-F$_{1}$ & 100 & AR1 & 82.0 & 0.0  & 0.3 & 0.3 (93.6) & 0.3 (90.7) \\
macro-F$_{1}$ & 100 & CS  & 82.0 & -0.2 & 1.0 & 1.0 (94.2) & 0.3 (42.3) \\
macro-F$_{1}$ & 200 & AR1 & 82.0 & 0.0  & 0.2 & 0.2 (95.0) & 0.2 (91.8) \\
macro-F$_{1}$ & 200 & CS  & 82.0 & 0.0  & 0.7 & 0.7 (94.8) & 0.2 (44.0) \\

micro-F$_{1}$ & 50  & AR1 & 82.0 & 0.0  & 0.4 & 0.4 (93.6) & 0.4 (91.8) \\
micro-F$_{1}$ & 50  & CS  & 82.0 & 0.0  & 1.3 & 1.3 (93.6) & 0.4 (43.5) \\
micro-F$_{1}$ & 100 & AR1 & 82.0 & 0.0  & 0.3 & 0.3 (93.4) & 0.3 (90.8) \\
micro-F$_{1}$ & 100 & CS  & 82.0 & 0.0  & 0.9 & 0.9 (94.2) & 0.3 (43.2) \\
micro-F$_{1}$ & 200 & AR1 & 82.0 & 0.0  & 0.2 & 0.2 (94.9) & 0.2 (91.9) \\
micro-F$_{1}$ & 200 & CS  & 82.0 & 0.0  & 0.6 & 0.7 (95.0) & 0.2 (44.8) \\

\midrule
\multicolumn{8}{l}{\textbf{Imbalanced:} $P(\text{Class 1})=0.47$, $P(\text{Class 2})=0.26$, $P(\text{Class 3})=0.27$}\\
\midrule
macro-F$_{1}$ & 50  & AR1 & 75.6 & 0.0  & 0.5 & 0.5 (94.7) & 0.5 (92.2) \\
macro-F$_{1}$ & 50  & CS  & 75.6 & -0.3 & 1.8 & 1.7 (91.2) & 0.5 (36.1) \\
macro-F$_{1}$ & 100 & AR1 & 75.6 & 0.0  & 0.4 & 0.4 (95.0) & 0.3 (92.7) \\
macro-F$_{1}$ & 100 & CS  & 75.6 & -0.2 & 1.3 & 1.2 (93.1) & 0.3 (39.2) \\
macro-F$_{1}$ & 200 & AR1 & 75.6 & 0.0  & 0.3 & 0.3 (94.7) & 0.2 (91.4) \\
macro-F$_{1}$ & 200 & CS  & 75.6 & -0.1 & 0.9 & 0.9 (93.7) & 0.2 (37.7) \\

micro-F$_{1}$ & 50  & AR1 & 78.0 & 0.0  & 0.6 & 0.5 (93.8) & 0.4 (85.2) \\
micro-F$_{1}$ & 50  & CS  & 78.0 & 0.0  & 2.2 & 2.1 (92.9) & 0.4 (29.2) \\
micro-F$_{1}$ & 100 & AR1 & 78.0 & 0.0  & 0.4 & 0.4 (94.8) & 0.3 (88.6) \\
micro-F$_{1}$ & 100 & CS  & 78.0 & 0.0  & 1.5 & 1.5 (94.0) & 0.3 (28.4) \\
micro-F$_{1}$ & 200 & AR1 & 78.0 & 0.0  & 0.3 & 0.3 (94.2) & 0.2 (87.4) \\
micro-F$_{1}$ & 200 & CS  & 78.0 & 0.0  & 1.1 & 1.1 (94.7) & 0.2 (28.3) \\
\bottomrule
\end{tabular}
\end{table}
Table~\ref{tab:sim_binary_rho08}--\ref{tab:sim_multiclass_rho08} summarizes the simulation results under strong within-cluster dependence ($\rho=0.8$). Across all configurations, the plug-in estimators show essentially no bias for any of the metrics considered.

The cluster-robust ASEs (ASE$_{\mathrm{rob}}$) closely track the ESE, even under strong dependence and moderate sample sizes. This suggests that the proposed sandwich variance estimator provides reliable uncertainty quantification in finite samples across a range of metrics, including both simple ratio-based measures and more nonlinear functionals such as MCC and macro-F$_{1}$.

The proposed cluster-robust intervals achieve coverage close to the nominal $95\%$ level. Some undercoverage appears at smaller sample sizes ($n=50$), especially under stronger dependence, but this diminishes as $n$ increases, with coverage stabilizing near the nominal level. This pattern is expected because the asymptotic approximation is driven primarily by the number of independent clusters. Thus, when the number of clusters is small, more conservative Type~I error control may be warranted.

In contrast, the naive standard error (ASE$_{\mathrm{naive}}$) that ignores within-cluster dependence substantially underestimates variability. As a result, the corresponding CIs exhibit severe undercoverage, particularly under CS, where coverage frequently falls in the $20\%$--$40\%$ range. Even under AR(1) dependence, the naive approach shows systematic undercoverage, though less extreme. Overall,  the impact of ignoring clustering varies with the dependence structure but is consistently problematic.

These results support the finite-sample validity of the proposed inference framework. The estimator is approximately unbiased, the cluster-robust variance estimator captures sampling variability well, and the resulting intervals attain near-nominal coverage. By contrast, ignoring within-cluster dependence can produce seriously misleading uncertainty quantification and invalid inference. Additional simulation results, including other performance metrics, a very small number of clusters, and an extreme class-imbalance setting, are provided in \ref{appendix:simulation}.

\subsection{Power and sample size approximations}
\subsubsection{Simulation design}

Having established the finite-sample validity of the proposed inference procedure under general clustered dependence structures, we next examine its practical implications for study design, focusing on power and sample size determination. To this end, we run an additional simulation in a clustered binary classification setting using the F$_{1}$ score.

For this design-oriented evaluation, we use a simpler data-generating mechanism than in the previous subsection. While the earlier simulations aim to validate asymptotic properties under flexible dependence structures, the present setting isolates the behavior of the power and sample size approximations in a controlled and computationally efficient setting. The random-effects construction provides a transparent way to induce within-cluster dependence while allowing direct control over signal strength and variability.

We generate clustered binary outcomes as follows: for cluster $i$ and observation $j$ with a fixed cluster size $m=100$,
\[
Y_{ij} \sim \mathrm{Bernoulli}(\pi_i), 
\qquad 
\pi_i = \expit(\gamma_0 + b_i),
\]
where $b_i \sim N(0,\sigma_b^2)$ with $\sigma_b = 0.8$ induces within-cluster dependence and $\gamma_0 = -0.2$ controls the marginal event rate. Conditional on $Y_{ij}$, we generate predictions according to prespecified sensitivity and specificity values, which determine the probabilities of correct classification for positive and negative outcomes, respectively.

In the superiority setting, we evaluate a single model with sensitivity $0.82$ and specificity $0.85$ using the F$_{1}$ score, with
\[
H_0: \theta \le \theta_0 
\qquad \text{versus} \qquad 
H_1: \theta > \theta_0,
\]
where $\theta_0 = 0.8$ serves as a benchmark performance level.

In the noninferiority setting, we evaluate two models on the same clustered data. The reference model uses sensitivity and specificity $(0.82, 0.85)$, while the candidate model uses $(0.84, 0.84)$, representing a small tradeoff between sensitivity and specificity. The hypotheses of interest are
\[
H_0: d \le -\Delta 
\qquad \text{versus} \qquad 
H_1: d > -\Delta,
\]
where $d = \theta^{(c)} - \theta^{(f)}$ denotes the difference in F$_{1}$ scores and $\Delta = 0.01$ specifies the noninferiority margin.

For each setting, we generate a pilot dataset ($n=2000$ clusters) to estimate the target metric and its asymptotic variance $V$ and $V_d$. We use this large pilot sample only to obtain stable reference values for evaluating the accuracy of the proposed sample size approximations in this simulation; it is not intended to reflect a practical requirement. In real applications, these quantities can be estimated from existing data, development datasets, or smaller pilot samples, provided they are reasonably representative of the target population. We then use these estimates to compute the required sample size $N_{\mathrm{req}}$ for target power $\in\{80\%, 90\%\}$ at $\alpha=0.05$ using the proposed asymptotic formulas.

We evaluate empirical power using $1000$ Monte Carlo replicates at three design points: $0.8N_{\mathrm{req}}$, $N_{\mathrm{req}}$, and $1.2N_{\mathrm{req}}$, where $N_{\mathrm{req}}$ is the sample size obtained from the asymptotic approximation with pilot-based variance estimates. These design points correspond to $80\%$, $100\%$, and $120\%$ of the proposed sample size and allow us to assess calibration accuracy.

\subsubsection{Results}

\begin{table}[tbp]\footnotesize
\centering
\caption{Pilot-based calibration of sample size and empirical power for superiority and noninferiority testing using the F$_{1}$ score under clustered data. Results are shown for target power levels of $80\%$ and $90\%$. The proposed sample size ($N_{\mathrm{req}}$) is obtained from the asymptotic approximation using pilot-based variance estimation. Empirical power is evaluated at $0.8N_{\mathrm{req}}$, $N_{\mathrm{req}}$, and $1.2N_{\mathrm{req}}$ to assess calibration accuracy.}\label{tab:power_validation}
\begin{tabular}{lccccc}
\toprule
Design point 
& \makecell{Target \\ power (\%)} 
& \makecell{Total \\ sample size} 
& \makecell{Number of \\ clusters} 
& \makecell{Empirical \\ power (\%)} \\
\midrule

\multicolumn{5}{l}{\textbf{Superiority}}\\
\midrule
80\% of required sample size  & 80 & 5,700  & 57  & 71.2 \\
\textbf{Required} sample size & \textbf{80} & 7,100  & 71  & \textbf{80.2} \\
120\% of required sample size & 80 & 8,600  & 86  & 85.8 \\
\midrule
80\% of required sample size  & 90 & 7,900  & 79  & 82.8 \\
\textbf{Required} sample size & \textbf{90} & 9,800  & 98  & \textbf{90.3} \\
120\% of required sample size & 90 & 11,800 & 118 & 92.7 \\

\midrule
\multicolumn{5}{l}{\textbf{Noninferiority}}\\
\midrule
80\% of required sample size  & 80 & 5,800  & 58  & 74.0 \\
\textbf{Required} sample size & \textbf{80} & 7,200  & 72  & \textbf{80.2} \\
120\% of required sample size & 80 & 8,700  & 87  & 86.7 \\
\midrule
80\% of required sample size  & 90 & 8,000  & 80  & 84.0 \\
\textbf{Required} sample size & \textbf{90} & 10,000 & 100 & \textbf{91.4} \\
120\% of required sample size & 90 & 12,000 & 120 & 93.3 \\

\bottomrule
\end{tabular}
\end{table}

Table~\ref{tab:power_validation} demonstrates that the proposed asymptotic sample size formulas provide accurate and practically useful guidance for study design across multiple target power levels. In both the superiority and noninferiority settings, empirical power at the formula-based sample size closely matches the nominal targets of $80\%$ and $90\%$, indicating that the first-order approximation performs well in finite samples.

Across all configurations, empirical power increases monotonically with sample size, as expected. Importantly, calibration remains accurate even at the higher target of $90\%$ power, which typically requires more precise variance estimation and is more sensitive to approximation error. This suggests that the proposed framework can be meaningfully utilized across a range of practically relevant design specifications. 

Overall, these findings indicate that the proposed asymptotic approximations, combined with pilot-based variance estimation, provide a reliable and computationally efficient approach for power analysis and sample size determination in clustered classification settings. The findings further suggest that the approach scales well across different target power levels, supporting its use in a broad set of practical study design scenarios.

\section{Application: Human Activity Recognition Dataset}\label{sec:application}
\subsection{Overview}

The following application illustrates how the proposed framework integrates pilot-based study design and cluster-robust inference in a realistic ML setting. We use the publicly available Human Activity Recognition (HAR) dataset, which consists of smartphone sensor measurements (561 features, including processed signals from accelerometers and gyroscopes) collected from multiple individuals performing six predefined activities: walking, walking upstairs, walking downstairs, sitting, standing, and laying; see \cite{anguita2013public} for details.

This dataset is well-suited for our purposes because it naturally exhibits a clustered structure: multiple observations are recorded per subject, inducing within-cluster dependence. Such dependence is common in modern ML applications, including wearable sensing, medical imaging, and longitudinal monitoring. However, it is often ignored in practice. Treating all observations as independent may lead to underestimated uncertainty, echoing the simulation results in Section~\ref{sec:simulation}.

Our objective here is not to achieve state-of-the-art predictive performance, but rather to demonstrate how the proposed framework enables principled statistical inference for performance metrics in the presence of clustered observations. For clarity, we organize the illustration into two stages: (i) \textbf{Stage~1}, where we design hypothesis tests using a pilot dataset aligned with a specific research objective, and (ii) \textbf{Stage~2}, where we conduct final validation with uncertainty quantification on an independent dataset.

\subsection{Implementation}

The full dataset contains observations from only $30$ clusters. Given this limitation, we randomly partition the clusters into three disjoint subsets: (i) a small \textbf{Stage~1} training set (2 clusters), (ii) a \textbf{Stage~1} testing set (3 clusters) used for pilot estimation, and (iii) a larger \textbf{Stage~2} testing set (25 clusters) used for final validation. We emphasize that this split is intentionally constrained and serves only as an illustration. In a proper validation study, training and testing stages should be fully separated and based on sufficiently large and representative datasets \citep{maleki2022generalizability}.

For simplicity, we consider two random forest classifiers trained on the \textbf{Stage~1} training set:
\begin{itemize}[leftmargin=1.5em, nosep]
    \item A reference (gold-standard) model with $100$ trees, and
    \item A candidate (proposed) model with $10$ trees.
\end{itemize}
The candidate model represents a computationally lighter alternative of practical interest. We evaluate its performance through (i) superiority testing against a prespecified benchmark and (ii) noninferiority testing relative to the reference model. All evaluations use the macro-F$_{1}$ score, which captures class-wise performance and is widely used in multiclass classification.

For \emph{superiority testing}, we consider
\[
H_0:\theta \le \theta_0
\quad \text{vs} \quad
H_1:\theta > \theta_0,
\]
where $\theta$ denotes the macro-F$_{1}$ of the candidate model and $\theta_0=0.755$ is a prespecified benchmark representing a practically meaningful performance level. 

For \emph{noninferiority testing}, we compare the candidate model $\theta^{(c)}$ to the reference model $\theta^{(f)}$ via the performance difference
\[
d = \theta^{(c)}-\theta^{(f)}
\]
and test
\[
H_0: d \le -\Delta
\quad \text{vs} \quad
H_1: d > -\Delta,
\]
where $\Delta=0.036$ represents the maximum acceptable loss in macro-F$_{1}$.

The choice of $\theta_0=0.755$ and $\Delta=0.036$ reflects the limited size of the dataset and is intended for illustration. In practice, these values should be specified based on domain knowledge and the underlying scientific objectives. We set the target power to $90\%$ and the significance level to $\alpha=0.05$ in this example for the illustrative purposes only.

\subsubsection{Stage 1: Pilot-based design}

The primary role of \textbf{Stage~1} is model fitting and initial performance assessment. Although training is the main objective, it is standard practice to further split the available data into training and testing subsets, as we do here. Additional steps such as hyperparameter tuning could be incorporated, but we omit them for simplicity, since our goal is to illustrate the inferential framework rather than optimize predictive performance.

We fit both models using the \textbf{Stage~1} training set ($n=2$) and evaluate their performance, along with the corresponding cluster-robust variance components, on the \textbf{Stage~1} testing set ($n=3$). We then use these pilot estimates to formulate hypothesis tests and determine the sample size required for reliable inference, following the design framework developed in Section~\ref{sec:method} and evaluated in Section~\ref{sec:simulation}.

Using the pilot estimates $\widehat\theta_1 \approx 0.786$ and $\widehat d_1 \approx -0.015$, together with the corresponding variance estimates $\widehat V \approx 0.933$ and $\widehat V_d \approx 0.521$, we compute the required sample sizes to achieve $90\%$ power at significance level $\alpha=0.05$ for both superiority and noninferiority testing. Based on the asymptotic formulas in Section~\ref{sec:method}, and assuming an average of approximately $\bar m =369$ observations per cluster (estimated from the training set), the required sample sizes are at least $23$ clusters for superiority testing and $28$ clusters for noninferiority testing.

Again, we emphasize that these relatively small sample size requirements are driven by the specific choices of $\theta_0$ and $\Delta$, which are intentionally calibrated to the limited \textbf{Stage 2} testing set size ($n=25$). In practice, these quantities should be carefully justified. To complement this analysis, we also construct power curves as functions of the number of clusters and visualize how statistical power varies with sample size in both superiority and noninferiority settings; see Figure~\ref{app:power}.

\begin{figure}[h]
\centering
\includegraphics[width=0.9\textwidth]{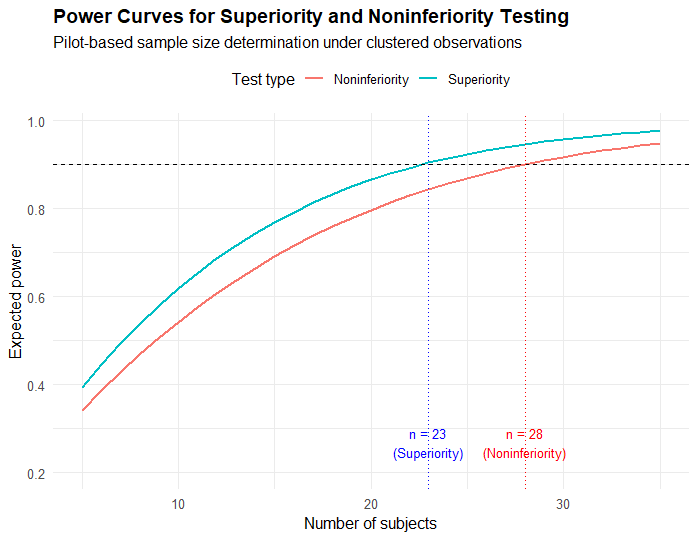}
\caption{Power curves for superiority and noninferiority testing based on pilot estimates under clustered observations. The horizontal dashed line indicates the target power of $90\%$. Vertical dotted lines denote the corresponding required numbers of clusters ($n=23$ for superiority and $n=28$ for noninferiority). The figure illustrates how statistical detectability depends on sample size.}\label{app:power}
\end{figure}

\subsubsection{Stage 2: Final validation and inference}

This stage corresponds to the final validation setting considered in the simulation studies, where inference is carried out using a fixed model and independent evaluation data. The primary role of \textbf{Stage~2} is formal validation and statistical inference using an independent dataset. In contrast to \textbf{Stage~1}, which focuses on model fitting and pilot-based study design, \textbf{Stage~2} serves as the final evaluation stage, where prespecified hypotheses are tested and uncertainty is properly quantified. No model training or tuning takes place at this stage.

We apply both the reference and candidate models, trained in \textbf{Stage~1}, to the \textbf{Stage~2} testing set and compute the corresponding performance metrics and cluster-robust variance estimates. We then use these quantities to conduct superiority and noninferiority tests under the hypotheses specified in the previous subsection; see Table~\ref{tab:app} for a summary.

\begin{table}[tbp]\footnotesize
\centering
\caption{Comparison of cluster-robust and naive inference on macro-F1 score $(\theta)$ in the HAR application. The naive approach treats all lower-level observations as independent, whereas the proposed method accounts for cluster-level dependence.}
\label{tab:app}
\begin{tabular}{llccccl}
\toprule
Method & Estimate & SE & \makecell{95\%\\ one-sided LB} & $p$-value & Conclusion \\
\midrule
\multicolumn{6}{l}{\textbf{Superiority testing:} ($H_0:\theta^{(c)} \le 0.755$)}\\
\midrule
Cluster-robust & $\widehat\theta=0.774$ & 0.020 & 0.74 & 0.181 & Fail to reject $H_0$ \\
Naive iid      & $\widehat\theta=0.774$ & 0.004 & 0.766 & $<0.001$ & Reject $H_0$ \\
\midrule
\multicolumn{6}{l}{\textbf{Noninferiority testing:} ($H_0: d \le -0.036$, where $d=\theta^{(c)}-\theta^{(f)}$)}\\
\midrule
Cluster-robust & $\widehat d=-0.072$ & 0.013 & -0.093 & 0.998 & Fail to reject $H_0$ \\
Naive iid      & $\widehat d=-0.072$ & 0.004 & -0.079 & $>$0.999 & Fail to reject $H_0$ \\
\bottomrule
\end{tabular}
\end{table}

For superiority testing, we evaluate whether the performance of the candidate model exceeds the prespecified benchmark $\theta_0=0.755$. The estimated macro-F$_{1}$ score is $\widehat\theta \approx 0.774$. Under the proposed cluster-robust inference procedure, the standard error is approximately $0.02$, yielding a 95\% CI of $[0.734,\,0.814]$ and a one-sided lower confidence bound of $0.74$. The corresponding $p$-value is $0.181$, so at the $\alpha=0.05$ significance level, we do not have sufficient evidence to conclude that the model exceeds the benchmark.

Next, we compare the candidate model to the reference model using the performance difference $d = \theta^{(c)}-\theta^{(f)}$ and a prespecified noninferiority margin of $\Delta=0.036$. The estimated difference is $\widehat d \approx -0.072$, indicating that the candidate model performs worse than the reference, which is expected given that the reference model uses a larger number of trees. The cluster-robust standard error is $0.013$, yielding a 95\% CI of $[-0.097,\,-0.047]$ and a one-sided lower bound of $-0.093$. The resulting $p$-value is $0.998$, so we fail to reject the null hypothesis.

For comparison, we also report results from a naive analysis that treats all lower-level observations as independent. In the HAR application, the naive approach produces substantially smaller standard errors ($0.004$ versus $0.02$) for the superiority test. The point estimate remains the same $\widehat\theta \approx 0.774$, but the CI is much narrower ($[0.765,\,0.782]$) and the $p$-value is below $<0.001$, overstating the strength of the evidence. A similar pattern appears in the  noninferiority test: the naive standard error ($0.004$ vs.\ $0.013$) produces a tighter CI ($[-0.080,\,-0.064]$ vs.\ $[-0.097,\,-0.047]$). These findings highlight that ignoring cluster-level dependence can materially distort uncertainty estimates and lead to misleading inferences.

Overall, this example demonstrates how the proposed framework integrates pilot-based study design and cluster-robust inference in practice. While point estimates are unchanged, properly accounting for within-cluster dependence has a substantial impact on uncertainty quantification and, consequently, on inferential conclusions. 

\section{Discussion}\label{sec:discussion}

In this work, we developed a unified inferential framework for evaluating prediction performance metrics under clustered data. By expressing performance measures as smooth functionals of confusion-matrix probabilities and combining this representation with cluster-adjusted sandwich variance estimation, the proposed approach enables asymptotically valid statistical inference for a broad class of metrics. This includes nonlinear measures such as the F$_{1}$ score and MCC, and applies in both binary and multiclass settings.

A key contribution is the explicit link between inference and study design. Beyond constructing CIs and conducting hypothesis tests, the framework also supports power calculations and sample size determination. These are based on first-order asymptotic approximations together with pilot-based variance estimates. This design-oriented perspective is particularly useful in applications where evaluation data are limited, expensive to collect, or exhibit structured dependence, and where conclusions must be drawn from relatively small samples.

The simulation studies suggested that the proposed estimators are approximately unbiased and that the cluster-adjusted variance estimator captures sampling variability well across a range of dependence structures and performance metrics. In contrast, naive approaches that ignore within-cluster dependence substantially underestimated variability, resulting in undercoverage and potentially misleading inference. We saw the same pattern in the real-data application in Section~\ref{sec:application}, where accounting for clustering altered the conclusion in a superiority test. These findings underscore that accounting for dependence is not just a technical refinement; it can materially affect conclusions of model evaluation.

Methodologically, the proposed framework builds on the general idea of sandwich variance estimation (e.g., \citep{liang1986longitudinal, colin2015practitioner, abadie2023should}), but differs in its target. Rather than focusing on regression parameters, it operates on functionals of confusion-matrix probabilities. The method covers a wide class of commonly used performance metrics without requiring metric-specific derivations.

Several limitations and directions for future work are worth noting. First, the asymptotic approximations rely on having a sufficiently large number of clusters and may break down when the number of clusters is small. In such cases, finite-sample corrections (e.g., \citep{windmeijer2005finite}) or resampling-based approaches like the cluster bootstrap (e.g., \citep{colin2015practitioner, mackinnon2017wild}) may improve performance. Although exact methods (e.g., \citep{clopper1934use, chan1999test}) would be desirable for ensuring nominal coverage, there is currently no general, model-free exact variance estimator for smooth performance metrics under arbitrary within-cluster dependence, to our knowledge.

Second, the current framework assumes independence across the top-level clusters. This allows for arbitrarily complex nested dependence within clusters, as long as clustering is performed at the highest independent level; for example, repeated measurements within visits within patients can be handled by treating the patient as the cluster. However, the framework does not directly accommodate nonnested, multi-way dependence, such as correlations induced jointly by visit time and clinical site. In such settings, extensions based on multi-way clustering may be appropriate \citep{cameron2011robust, colin2015practitioner}. More general forms of cross-cluster dependence that violate Assumption \ref{ass:mc-ind} remain an open problem. In many biomedical validation studies, however, this limitation may be less restrictive because a natural independent unit, such as the patient, can often be defined.

Third, the framework is restricted to metrics that can be expressed as smooth functions of confusion-matrix probabilities. Measures that fall outside this class, such as the intraclass correlation coefficient \citep{bartko1966intraclass, muller1994critical}, would require different methodological study.

Finally, extending the framework to settings involving data reuse, such as cross-validation \citep{stone1974cross, hastie2009elements} or adaptive model selection \citep{berk2013valid}, will require additional methodological development. Accounting for these extra layers of dependence while maintaining computational and analytical tractability is an important direction for future research.

Overall, this work provides a general and practical foundation for statistically principled evaluation of prediction performance under clustered data. It helps close the gap between statistical methodology and evaluation practice in biomedical ML. By combining asymptotic theory with design considerations, the proposed approach offers a coherent approach to more reliable evaluation, particularly in settings where dependence is inherent and cannot be ignored.

\section*{Declaration of Generative AI and AI-assisted technologies in the writing process}
During the preparation of this work, the authors employed ChatGPT to assist with linguistic refinement and R code development. All AI-assisted content was subsequently reviewed and edited by the authors, who bear full responsibility for the integrity and accuracy of the published material.

\section*{Disclaimer}
The views expressed in this article should not be construed to represent those of U.S. Food and Drug Administration.

\section*{Conflict of Interest}
The authors declare no conflicts of interest.



\section*{Data Availability}
Data used in this study were obtained from the UCI Machine Learning Repository (\href{https://uci-ics-mlr-prod.aws.uci.edu/}{https://uci-ics-mlr-prod.aws.uci.edu/}). HAR dataset is publicly available by searching `Human Activity Recognition Using Smartphones' in the repository. 
The R code used in the simulation and application is available at \href{https://github.com/tkh5956/Testing_Prediction}{https://github.com/tkh5956/Testing$\_$Prediction}.
\appendix
\section{Proof of Theorems}\label{appendix:proof}

\subsection{Useful Lemma}

Lemma~\ref{lem:clt-p} will be used repeatedly in the proofs that follow. Its proof follows standard arguments based on the Lyapunov central limit theorem for triangular arrays, the Cram\'er--Wold device, and the multivariate delta method; see, for example, \citet{van2000asymptotic}. Throughout, $C(\bm a,M)$ denotes a finite positive constant that may depend on $\bm a$ and $M$, and whose value may change from line to line. We write $a_n \asymp b_n$ to mean that $a_n/b_n$ is bounded away from zero and infinity.
\begin{lemma}[CLT for empirical confusion probabilities]\label{lem:clt-p}
Under Assumptions~\ref{ass:mc-ind}-\ref{ass:mc-cov},
\[
\sqrt{N_n}\bigl(\widehat{\bm p}_n-\bm p_n\bigr)
\dto
\mathcal{N}(\bm 0,\bm \Omega_r).
\]
In particular,
\[
\widehat{\bm p}_n-\bm p_n=\Op(N_n^{-1/2}).
\]
\end{lemma}

\begin{proof}
Recall that
\[
\widehat{\bm p}_n-\bm p_n
=
\frac{1}{N_n}\sum_{i=1}^n \bm U_{ni},
\qquad
\bm U_{ni}=\bm S_{ni}-m_{ni}\bm p_n,
\]
where $\E(\bm S_{ni}) = m_{ni}\bm p_n$ by Assumption~\ref{ass:mc-homog}. Thus,
\[
\sqrt{N_n}\bigl(\widehat{\bm p}_n-\bm p_n\bigr)
=
\frac{1}{\sqrt{N_n}}\sum_{i=1}^n \bm U_{ni}.
\]

To prove the multivariate central limit theorem, fix any \(\bm a\in\mathbb R^{r^2}\) and define
\[
X_{ni}=\bm a^\top \bm U_{ni},
\qquad i=1,\dots,n.
\]
By Assumption~\ref{ass:mc-ind}, the variables \(X_{ni}\) are independent across \(i\) for each \(n\). Moreover,
\[
\sum_{i=1}^n X_{ni}
=
\bm a^\top \sum_{i=1}^n \bm U_{ni}
=
N_n\,\bm a^\top(\widehat{\bm p}_n-\bm p_n).
\]

Let
\[
s_n^2=\sum_{i=1}^n \Var(X_{ni}).
\]
Then
\[
\frac{s_n^2}{N_n}
=
\bm a^\top
\left\{
\frac{1}{N_n}\sum_{i=1}^n \Var(\bm U_{ni})
\right\}
\bm a
=
\bm a^\top \bm \Omega_{r,n}\bm a
\to
\bm a^\top \bm \Omega_r\bm a
\]
by Assumption~\ref{ass:mc-cov}.

First, consider the case \(\bm a^\top \bm \Omega_r\bm a>0\). Since each \(\bm Z_{nij}\) is a one-hot vector, each component of \(\bm S_{ni}\) takes a value in \(\{0,1,\dots,m_{ni}\}\). It follows that each component of
\[
\bm U_{ni}=\bm S_{ni}-m_{ni}\bm p_n
\]
is bounded in absolute value by \(m_{ni}\). By Assumption~\ref{ass:mc-bound},
\[
\|\bm U_{ni}\|_\infty \le m_{ni}\le M.
\]
Therefore,
\[
|X_{ni}|=|\bm a^\top \bm U_{ni}|
\le
\|\bm a\|_1\|\bm U_{ni}\|_\infty
\le
M\|\bm a\|_1
\]
uniformly in \(n\) and \(i\). Then,
\[
\E|X_{ni}|^4\le C(\bm a,M)<\infty
\]
uniformly in \(n\) and \(i\), so
\[
\sum_{i=1}^n \E|X_{ni}|^4 \le n\,C(\bm a,M).
\]
Since \(N_n/n\to \bar m\in(0,M]\), we have \(N_n\asymp n\). If \(\bm a^\top\bm\Omega_r\bm a>0\), then \(s_n^2\asymp N_n\asymp n\), so \(s_n^4\asymp n^2\). Consequently,
\[
\frac{\sum_{i=1}^n \E|X_{ni}|^4}{s_n^4}\to 0.
\]
Thus, the Lyapunov condition with exponent \(4\) holds. By the Lyapunov CLT for triangular arrays,
\[
\frac{\sum_{i=1}^n X_{ni}}{s_n}\dto N(0,1).
\]
Equivalently,
\[
\frac{
\bm a^\top \sqrt{N_n}(\widehat{\bm p}_n-\bm p_n)
}{
\sqrt{\bm a^\top \bm \Omega_{r,n}\bm a}
}
\dto N(0,1).
\]
Since \(\bm a^\top \bm \Omega_{r,n}\bm a\to \bm a^\top \bm \Omega_r\bm a\), Slutsky's theorem yields
\[
\bm a^\top \sqrt{N_n}(\widehat{\bm p}_n-\bm p_n)
\dto
N\!\left(0,\bm a^\top \bm \Omega_r\bm a\right)
\]
for every \(\bm a\) such that \(\bm a^\top \bm \Omega_r\bm a>0\).

Now suppose \(\bm a^\top \bm \Omega_r\bm a=0\). Then
\[
\frac{1}{N_n}\sum_{i=1}^n \Var(X_{ni})
=
\bm a^\top \bm \Omega_{r,n}\bm a
\to 0.
\]
Hence,
\[
\Var\!\left(
\bm a^\top \sqrt{N_n}(\widehat{\bm p}_n-\bm p_n)
\right)
=
\Var\!\left(
\frac{1}{\sqrt{N_n}}\sum_{i=1}^n X_{ni}
\right)
=
\frac{1}{N_n}\sum_{i=1}^n \Var(X_{ni})
\to 0.
\]
By Chebyshev's inequality,
\[
\bm a^\top \sqrt{N_n}(\widehat{\bm p}_n-\bm p_n)\pto 0,
\]
which corresponds to a degenerate normal distribution.

Since this holds for every fixed \(\bm a\in\mathbb R^{r^2}\), the Cram\'er-Wold device implies
\[
\bm a^\top \sqrt{N_n}(\widehat{\bm p}_n-\bm p_n)
\dto
N\!\left(0,\bm a^\top \bm \Omega_r\bm a\right),
\]
and 
\[
\sqrt{N_n}\bigl(\widehat{\bm p}_n-\bm p_n\bigr)
\dto
\mathcal N(\bm 0,\bm \Omega_r).
\]

The final statement follows immediately, since convergence in distribution implies boundedness in probability:
\[
\sqrt{N_n}\bigl(\widehat{\bm p}_n-\bm p_n\bigr)=\Op(1),
\]
and therefore
\[
\widehat{\bm p}_n-\bm p_n=\Op(N_n^{-1/2}).
\]
\end{proof}

\subsection{Proof of Theorem \ref{thm:sandwich}}
Recall that under Assumption \ref{ass:mc-homog},
\[
\E(\bm Z_{nij})=\bm p_n
\qquad\text{for all } i=1,\dots,n,\ j=1,\dots,m_{ni},
\]
so that
\[
\E(\bm S_{ni}) = m_{ni}\bm p_n.
\]
Hence
\[
\bm U_{ni}=\bm S_{ni}-m_{ni}\bm p_n,
\qquad
\widehat{\bm U}_{ni}=\bm S_{ni}-m_{ni}\widehat{\bm p}_n.
\]

Define the infeasible oracle estimator
\[
\widetilde{\bm \Omega}_{r,n}
=
\frac{1}{N_n}\sum_{i=1}^n \bm U_{ni}\bm U_{ni}^\top.
\]
Since \(\bm \Omega_{r,n}\to\bm \Omega_r\) by Assumption \ref{ass:mc-cov}, it suffices to show that
\[
\widetilde{\bm \Omega}_{r,n} - \bm \Omega_{r,n}\pto 0,
\qquad
\widehat{\bm \Omega}_{r,n} - \widetilde{\bm \Omega}_{r,n}\pto 0.
\]

We first show that the oracle estimator is consistent for \(\bm \Omega_{r,n}\). Fix any entry \((k,\ell)\) of the matrix \(\widetilde{\bm \Omega}_{r,n}\) and define
\[
Y_{ni}^{(k\ell)} = U_{ni,k}U_{ni,\ell}.
\]
By Assumption \ref{ass:mc-ind}, the variables \(Y_{ni}^{(k\ell)}\) are independent across \(i\) for each row \(n\). Since each \(\bm Z_{nij}\) is a one-hot vector, each component of \(\bm S_{ni}\) lies in \(\{0,1,\dots,m_{ni}\}\), and therefore each component of
\[
\bm U_{ni}=\bm S_{ni}-m_{ni}\bm p_n
\]
is bounded in absolute value by \(m_{ni}\). By Assumption \ref{ass:mc-bound},
\[
|U_{ni,k}| \le m_{ni}\le M,
\qquad
|U_{ni,\ell}| \le m_{ni}\le M,
\]
so
\[
|Y_{ni}^{(k\ell)}| \le M^2
\]
uniformly in \(n\) and \(i\). Therefore
\[
\Var\!\left(
\frac{1}{N_n}\sum_{i=1}^n \{Y_{ni}^{(k\ell)}-\E (Y_{ni}^{(k\ell)})\}
\right)
=
\frac{1}{N_n^2}\sum_{i=1}^n \Var(Y_{ni}^{(k\ell)})
\le
\frac{nC}{N_n^2}
\to 0,
\]
for some finite constant $C$. Since \(N_n/n\to \bar m>0\), it follows that \(n/N_n^2 \to 0\). Hence, by Chebyshev's inequality,
\[
\frac{1}{N_n}\sum_{i=1}^n \{Y_{ni}^{(k\ell)}-\E (Y_{ni}^{(k\ell)})\}\pto 0.
\]
Since this holds for every fixed entry \((k,\ell)\),
\[
\widetilde{\bm \Omega}_{r,n}
-
\frac{1}{N_n}\sum_{i=1}^n \E(\bm U_{ni}\bm U_{ni}^\top)
\pto 0,
\]
but
\[
\frac{1}{N_n}\sum_{i=1}^n \E(\bm U_{ni}\bm U_{ni}^\top)
=
\frac{1}{N_n}\sum_{i=1}^n \Var(\bm U_{ni})
=
\bm \Omega_{r,n}.
\]
Therefore
\[
\widetilde{\bm \Omega}_{r,n}-\bm \Omega_{r,n}\pto 0.
\]

We next replace the oracle centering by the plug-in centering. Since
\[
\widehat{\bm U}_{ni}
=
\bm S_{ni}-m_{ni}\widehat{\bm p}_n
=
\bm U_{ni} - m_{ni}(\widehat{\bm p}_n-\bm p_n),
\]
we have
\begin{align*}
\widehat{\bm U}_{ni}\widehat{\bm U}_{ni}^\top
&=
\bm U_{ni}\bm U_{ni}^\top
- m_{ni}\bm U_{ni}(\widehat{\bm p}_n-\bm p_n)^\top
- m_{ni}(\widehat{\bm p}_n-\bm p_n)\bm U_{ni}^\top \\
&\quad
+ m_{ni}^2(\widehat{\bm p}_n-\bm p_n)(\widehat{\bm p}_n-\bm p_n)^\top.
\end{align*}
Summing over \(i\) and dividing by \(N_n\) yields
\begin{align*}
\widehat{\bm \Omega}_{r,n}-\widetilde{\bm \Omega}_{r,n}
&=
-\frac{1}{N_n}\sum_{i=1}^n m_{ni}\bm U_{ni}(\widehat{\bm p}_n-\bm p_n)^\top \\
&\quad
-\frac{1}{N_n}\sum_{i=1}^n m_{ni}(\widehat{\bm p}_n-\bm p_n)\bm U_{ni}^\top \\
&\quad
+\frac{1}{N_n}\sum_{i=1}^n m_{ni}^2(\widehat{\bm p}_n-\bm p_n)(\widehat{\bm p}_n-\bm p_n)^\top.
\end{align*}

To bound the cross terms, note that for each fixed coordinate \(k\),
\[
\frac{1}{N_n}\sum_{i=1}^n m_{ni}U_{ni,k}
\]
is an average of independent mean-zero variables across \(i\). Since \(m_{ni}\le M\) and \(|U_{ni,k}|\le m_{ni}\le M\), the variables \(m_{ni}U_{ni,k}\) are uniformly bounded by \(M^2\). Moreover,
\[
\frac{1}{N_n}\sum_{i=1}^n \Var(m_{ni}U_{ni,k})
\le
\frac{M^2}{N_n}\sum_{i=1}^n \Var(U_{ni,k})
= O(1).
\]
Hence, by an argument analogous to that used in Lemma~\ref{lem:clt-p},
\[
\frac{1}{N_n}\sum_{i=1}^n m_{ni}U_{ni,k}
=
O_p(N_n^{-1/2}).
\]
Since this holds for each coordinate \(k\), it follows that
\[
\frac{1}{N_n}\sum_{i=1}^n m_{ni}\bm U_{ni}
=
O_p(N_n^{-1/2}).
\]
By Lemma \ref{lem:clt-p},
\[
\widehat{\bm p}_n-\bm p_n = \Op(N_n^{-1/2}),
\]
therefore each cross term is
\[
\Op(N_n^{-1/2})\cdot \Op(N_n^{-1/2})
=
\Op(N_n^{-1})
=
o_p(1).
\]

For the quadratic term,
\[
\left\|
\frac{1}{N_n}\sum_{i=1}^n
m_{ni}^2(\widehat{\bm p}_n-\bm p_n)(\widehat{\bm p}_n-\bm p_n)^\top
\right\|
\le
\frac{nM^2}{N_n}\|\widehat{\bm p}_n-\bm p_n\|^2.
\]
Since \(n/N_n=\Op(1)\) and \(\|\widehat{\bm p}_n-\bm p_n\|^2=\Op(N_n^{-1})\),
\[
\frac{nM^2}{N_n}\|\widehat{\bm p}_n-\bm p_n\|^2
=
\Op(1)\cdot \Op(N_n^{-1})
=
o_p(1).
\]
Hence
\[
\widehat{\bm \Omega}_{r,n}-\widetilde{\bm \Omega}_{r,n}\pto 0.
\]

Combining the two parts,
\[
\widehat{\bm \Omega}_{r,n}-\bm \Omega_{r,n}\pto 0.
\]
Since \(\bm \Omega_{r,n}\to \bm \Omega_r\) by Assumption \ref{ass:mc-cov}, Slutsky's theorem gives
\[
\widehat{\bm \Omega}_{r,n}\pto \bm \Omega_r.
\]
This proves the theorem.

\subsection{Proof of Theorem \ref{thm:delta-general}}

The asymptotic normality of \(\widehat{\bm p}_n\) has been established by Lemma \ref{lem:clt-p}. We now apply the multivariate delta method. By Assumption \ref{ass:smooth}, \(g\) is continuously differentiable on an open neighborhood of \(\bm p\), and by Assumption \ref{ass:mc-first}, \(\bm p_n\to \bm p\). By a first-order Taylor expansion of \(g\) around \(\bm p_n\),
\[
\sqrt{N_n}\{g(\widehat{\bm p}_n)-g(\bm p_n)\}
=
\nabla g(\bm p_n)^\top \sqrt{N_n}(\widehat{\bm p}_n-\bm p_n) + o_p(1).
\]
Because \(\nabla g(\bm p_n)\to \nabla g(\bm p)\), Slutsky's theorem yields
\[
\sqrt{N_n}(\widehat\theta_n-\theta_n)
=
\sqrt{N_n}\{g(\widehat{\bm p}_n)-g(\bm p_n)\}
\dto
N(0,V),
\]
where
\[
V=\nabla g(\bm p)^\top \bm \Omega_r \nabla g(\bm p).
\]
This proves the theorem.

\subsection{Proof of Theorem \ref{thm:pair}}

For each cluster \(i=1,\dots,n\), let
\[
\bm S_{ni}^{(c)}
\qquad\text{and}\qquad
\bm S_{ni}^{(f)}
\]
denote the cluster-level confusion-count vectors for the candidate and reference models, respectively. Under Assumption~\ref{ass:mc-homog} applied separately to the candidate and reference models,
\[
\E(\bm S_{ni}^{(c)}) = m_{ni}\bm p_n^{(c)},
\qquad
\E(\bm S_{ni}^{(f)}) = m_{ni}\bm p_n^{(f)}.
\]
Define
\[
\bm U_{ni}^{(c)}
=
\bm S_{ni}^{(c)}-m_{ni}\bm p_n^{(c)},
\qquad
\bm U_{ni}^{(f)}
=
\bm S_{ni}^{(f)}-m_{ni}\bm p_n^{(f)}.
\]
Now stack them as
\[
\bar{\bm U}_{ni}
=
\begin{pmatrix}
\bm U_{ni}^{(c)}\\
\bm U_{ni}^{(f)}
\end{pmatrix}
\in \R^{2r^2}.
\]
Then
\[
\sum_{i=1}^n \bar{\bm U}_{ni}
=
N_n
\begin{pmatrix}
\widehat{\bm p}_n^{(c)}-\bm p_n^{(c)}\\
\widehat{\bm p}_n^{(f)}-\bm p_n^{(f)}
\end{pmatrix}.
\]

Let
\[
\bar{\bm \Omega}_{r,n}
=
\frac{1}{N_n}\sum_{i=1}^n \Var(\bar{\bm U}_{ni}).
\]
By Assumption \ref{ass:mc-cov-pair},
\[
\bar{\bm \Omega}_{r,n}
=
\begin{pmatrix}
\bm \Omega_{r,n}^{(c)} & \bm \Omega_{r,n}^{(cf)}\\
\bm \Omega_{r,n}^{(fc)} & \bm \Omega_{r,n}^{(f)}
\end{pmatrix}
\to
\bar{\bm \Omega}_r
=
\begin{pmatrix}
\bm \Omega_r^{(c)} & \bm \Omega_r^{(cf)}\\
\bm \Omega_r^{(fc)} & \bm \Omega_r^{(f)}
\end{pmatrix}.
\]

Fix any vector \(\bar{\bm a}\in\R^{2r^2}\) and define
\[
X_{ni}
=
\bar{\bm a}^\top \bar{\bm U}_{ni},
\qquad i=1,\dots,n.
\]
By Assumption \ref{ass:mc-ind}, the clusters are independent across \(i\), so the scalar variables \(X_{ni}\) are independent across \(i\) for each row \(n\). Moreover, since each component of \(\bm U_{ni}^{(c)}\) and \(\bm U_{ni}^{(f)}\) is bounded in absolute value by \(m_{ni}\le M\), each component of \(\bar{\bm U}_{ni}\) is uniformly bounded by \(M\). Therefore
\[
|X_{ni}|
\le
\|\bar{\bm a}\|_1\|\bar{\bm U}_{ni}\|_\infty
\le
M\|\bar{\bm a}\|_1
\]
uniformly in \(n\) and \(i\). In particular,
\[
\E|X_{ni}|^4 \le C(\bar{\bm a},M)<\infty
\]
uniformly in \(n\) and \(i\), so
\[
\sum_{i=1}^n \E|X_{ni}|^4 \le n\,C(\bar{\bm a},M).
\]

Let
\[
s_n^2=\sum_{i=1}^n \Var(X_{ni}).
\]
Then
\[
\frac{s_n^2}{N_n}
=
\bar{\bm a}^\top
\left\{
\frac{1}{N_n}\sum_{i=1}^n \Var(\bar{\bm U}_{ni})
\right\}
\bar{\bm a}
=
\bar{\bm a}^\top \bar{\bm \Omega}_{r,n}\bar{\bm a}
\to
\bar{\bm a}^\top \bar{\bm \Omega}_r\bar{\bm a}.
\]

If \(\bar{\bm a}^\top \bar{\bm \Omega}_r \bar{\bm a}>0\), then \(s_n^2\) is of order \(N_n\), hence of order \(n\), so \(s_n^4\) is of order \(n^2\). Therefore
\[
\frac{\sum_{i=1}^n \E|X_{ni}|^4}{s_n^4}\to 0,
\]
which is the Lyapunov condition with exponent \(4\). Hence
\[
\frac{\sum_{i=1}^n X_{ni}}{s_n}\dto N(0,1).
\]
Equivalently,
\[
\frac{
\bar{\bm a}^\top
\sqrt{N_n}
\begin{pmatrix}
\widehat{\bm p}_n^{(c)}-\bm p_n^{(c)}\\
\widehat{\bm p}_n^{(f)}-\bm p_n^{(f)}
\end{pmatrix}
}{
\sqrt{\bar{\bm a}^\top \bar{\bm \Omega}_{r,n}\bar{\bm a}}
}
\dto N(0,1).
\]
Since
\[
\bar{\bm a}^\top \bar{\bm \Omega}_{r,n}\bar{\bm a}
\to
\bar{\bm a}^\top \bar{\bm \Omega}_r\bar{\bm a},
\]
Slutsky's theorem yields
\[
\bar{\bm a}^\top
\sqrt{N_n}
\begin{pmatrix}
\widehat{\bm p}_n^{(c)}-\bm p_n^{(c)}\\
\widehat{\bm p}_n^{(f)}-\bm p_n^{(f)}
\end{pmatrix}
\dto
N\!\left(0,\bar{\bm a}^\top \bar{\bm \Omega}_r\bar{\bm a}\right).
\]

If \(\bar{\bm a}^\top \bar{\bm \Omega}_r \bar{\bm a}=0\), then
\[
\frac{1}{N_n}\sum_{i=1}^n \Var(X_{ni})
=
\bar{\bm a}^\top \bar{\bm \Omega}_{r,n}\bar{\bm a}
\to 0,
\]
so
\[
\Var\!\left(
\bar{\bm a}^\top
\sqrt{N_n}
\begin{pmatrix}
\widehat{\bm p}_n^{(c)}-\bm p_n^{(c)}\\
\widehat{\bm p}_n^{(f)}-\bm p_n^{(f)}
\end{pmatrix}
\right)
=
\frac{1}{N_n}\sum_{i=1}^n \Var(X_{ni})
\to 0.
\]
Hence, by Chebyshev's inequality,
\[
\bar{\bm a}^\top
\sqrt{N_n}
\begin{pmatrix}
\widehat{\bm p}_n^{(c)}-\bm p_n^{(c)}\\
\widehat{\bm p}_n^{(f)}-\bm p_n^{(f)}
\end{pmatrix}
\pto 0.
\]

Since this holds for every fixed \(\bm a\in\mathbb R^{r^2}\), the Cram\'er--Wold device implies
\[
\bar{\bm a}^\top
\sqrt{N_n}
\begin{pmatrix}
\widehat{\bm p}_n^{(c)}-\bm p_n^{(c)}\\
\widehat{\bm p}_n^{(f)}-\bm p_n^{(f)}
\end{pmatrix}
\dto
N\!\left(0,\bar{\bm a}^\top \bar{\bm \Omega}_r\bar{\bm a}\right),
\]
and
\[
\sqrt{N_n}
\begin{pmatrix}
\widehat{\bm p}_n^{(c)}-\bm p_n^{(c)}\\
\widehat{\bm p}_n^{(f)}-\bm p_n^{(f)}
\end{pmatrix}
\dto
N\!\left(
\bm 0,
\begin{pmatrix}
\bm \Omega_r^{(c)} & \bm \Omega_r^{(cf)}\\
\bm \Omega_r^{(fc)} & \bm \Omega_r^{(f)}
\end{pmatrix}
\right).
\]

Define 
\[
h(\bm u,\bm v)=g(\bm u)-g(\bm v),
\qquad
(\bm u,\bm v)\in\R^{r^2}\times\R^{r^2}.
\]
By Assumption~\ref{ass:smooth}, \(h\) is continuously differentiable in a neighborhood of \((\bm p^{(c)},\bm p^{(f)})\). Then
\[
d_n=h(\bm p_n^{(c)},\bm p_n^{(f)}),
\qquad
\widehat d_n=h(\widehat{\bm p}_n^{(c)},\widehat{\bm p}_n^{(f)}).
\]
The gradient of \(h\) at \((\bm p^{(c)},\bm p^{(f)})\) is
\[
\nabla h(\bm p^{(c)},\bm p^{(f)})
=
\begin{pmatrix}
\nabla g(\bm p^{(c)})\\
-\nabla g(\bm p^{(f)})
\end{pmatrix}.
\]
By the multivariate delta method,
\[
\sqrt{N_n}(\widehat d_n-d_n)\dto N(0,V_d),
\]
where
\begin{align*}
V_d
&=
\nabla h(\bm p^{(c)},\bm p^{(f)})^\top
\bar{\bm \Omega}_r
\nabla h(\bm p^{(c)},\bm p^{(f)}) \\
&=
\nabla g(\bm p^{(c)})^\top \bm \Omega_r^{(c)} \nabla g(\bm p^{(c)})
+
\nabla g(\bm p^{(f)})^\top \bm \Omega_r^{(f)} \nabla g(\bm p^{(f)}) \\
&\qquad
-
\nabla g(\bm p^{(c)})^\top \bm \Omega_r^{(cf)} \nabla g(\bm p^{(f)})
-
\nabla g(\bm p^{(f)})^\top \bm \Omega_r^{(fc)} \nabla g(\bm p^{(c)}).
\end{align*}
Since \(\bm \Omega_r^{(fc)}=(\bm \Omega_r^{(cf)})^\top\), the last two terms are equal, and therefore
\[
V_d
=
\nabla g(\bm p^{(c)})^\top \bm \Omega_r^{(c)} \nabla g(\bm p^{(c)})
+
\nabla g(\bm p^{(f)})^\top \bm \Omega_r^{(f)} \nabla g(\bm p^{(f)})
-
2\,\nabla g(\bm p^{(c)})^\top \bm \Omega_r^{(cf)} \nabla g(\bm p^{(f)}).
\]
This proves the theorem.

\section{Additional Simulation Results}\label{appendix:simulation}

Tables~\ref{tab:sim_binary_rho08_accf1}--\ref{tab:sim_multiclass_rho05} present additional simulation results for the binary and multiclass settings considered in the main manuscript, across different dependence structures and class distributions. Overall, the cluster-robust estimator maintains near-nominal coverage, whereas the naive estimator tends to underestimate variability and exhibits substantial undercoverage, particularly under stronger within-cluster dependence.

We further considered a challenging setting with only \(n=25\) independent clusters, holding everything else fixed. As shown in Tables~\ref{tab:sim_binary_balanced_n25}--\ref{tab:sim_multiclass_n25}, the cluster-robust estimator is still much better calibrated than the naive estimator, but coverage is less reliable than with larger $n$, especially under strong CS dependence and for nonlinear or imbalance-sensitive metrics. These results support the use of more conservative Type~I error control when such small studies are unavoidable.

We also examined an extreme class-imbalance scenario with \(1\%\) positive-class prevalence; see Tables~\ref{tab:sim_binary_rare_moderate_f1mcc} and~\ref{tab:sim_binary_rare_strong_f1mcc}. 
We report \(F_1\) and MCC as illustrative positive-class-related metrics, since metrics dominated by the negative class, such as accuracy and specificity, remain comparatively stable. Because the data-generating mechanisms retain the same sensitivity and specificity as in the main simulations, the rare-event setting produces many false positives relative to true positives, leading to small true values of \(F_1\) and MCC. These metrics are unstable in finite samples, especially under strong CS dependence and smaller \(n\), although performance improves as the number of independent clusters increases. Thus, extreme class imbalance can degrade normal-approximation-based inference even with cluster-robust variance estimation, and reliable approximation may require more independent clusters, potentially aided by many observations within each cluster.

\begin{table}[h]\footnotesize
\centering
\caption{Simulation results in the binary setting under stronger within-cluster dependence (\(\rho=0.8\)) for accuracy and \(F_1\). ESE indicates empirical standard error. ASE$_{\mathrm{rob}}$ (CP$_{\mathrm{rob}}$) and ASE$_{\mathrm{naive}}$ (CP$_{\mathrm{naive}}$) denote cluster-robust and naive asymptotic standard errors with corresponding coverage probabilities. True performance values, bias, and all SEs are reported on the \(\times 100\) scale, and CP is reported as a percentage (ideally 95\%).}
\label{tab:sim_binary_rho08_accf1}
\begin{tabular}{lccrrrrr}
\toprule
Metric & $n$ & Corr. & True & Bias & ESE & ASE$_{\mathrm{rob}}$ (CP) & ASE$_{\mathrm{naive}}$ (CP) \\
\midrule
\multicolumn{8}{l}{\textbf{Balanced:} $P(\text{Class 1})=0.5$, $P(\text{Class 2})=0.5$}\\
\midrule
Accuracy & 50  & AR1 & 70.0 & 0.0  & 0.7 & 0.6 (93.0) & 0.5 (82.9) \\
Accuracy & 50  & CS  & 70.0 & 0.0  & 3.3 & 3.2 (94.2) & 0.5 (19.6) \\
Accuracy & 100 & AR1 & 70.0 & 0.0  & 0.5 & 0.5 (94.8) & 0.3 (83.8) \\
Accuracy & 100 & CS  & 70.0 & -0.1 & 2.3 & 2.3 (94.6) & 0.3 (19.8) \\
Accuracy & 200 & AR1 & 70.0 & 0.0  & 0.3 & 0.3 (94.8) & 0.2 (83.2) \\
Accuracy & 200 & CS  & 70.0 & 0.1  & 1.6 & 1.6 (94.7) & 0.2 (24.1) \\

F$_{1}$ & 50  & AR1 & 70.0 & 0.0  & 1.0 & 0.9 (93.8) & 0.5 (71.4) \\
F$_{1}$ & 50  & CS  & 70.0 & -0.3 & 4.5 & 4.4 (94.2) & 0.5 (17.1) \\
F$_{1}$ & 100 & AR1 & 70.0 & 0.0  & 0.7 & 0.7 (95.2) & 0.4 (72.0) \\
F$_{1}$ & 100 & CS  & 70.0 & -0.3 & 3.2 & 3.2 (94.3) & 0.4 (18.2) \\
F$_{1}$ & 200 & AR1 & 70.0 & 0.0  & 0.5 & 0.5 (94.6) & 0.3 (71.0) \\
F$_{1}$ & 200 & CS  & 70.0 & 0.0  & 2.3 & 2.2 (94.0) & 0.3 (18.6) \\

\midrule
\multicolumn{8}{l}{\textbf{Imbalanced:} $P(\text{Class 1})=0.8$, $P(\text{Class 2})=0.2$}\\
\midrule
Accuracy & 50  & AR1 & 88.0 & 0.0  & 0.5 & 0.4 (93.4) & 0.3 (81.3) \\
Accuracy & 50  & CS  & 88.0 & 0.0  & 1.9 & 1.9 (93.3) & 0.3 (27.0) \\
Accuracy & 100 & AR1 & 88.0 & 0.0  & 0.3 & 0.3 (93.6) & 0.2 (83.4) \\
Accuracy & 100 & CS  & 88.0 & 0.0  & 1.3 & 1.3 (95.0) & 0.2 (24.1) \\
Accuracy & 200 & AR1 & 88.0 & 0.0  & 0.2 & 0.2 (95.0) & 0.2 (84.1) \\
Accuracy & 200 & CS  & 88.0 & 0.0  & 1.0 & 0.9 (93.8) & 0.2 (26.9) \\

F$_{1}$ & 50  & AR1 & 72.7 & -0.1 & 1.1 & 1.0 (92.9) & 0.8 (83.2) \\
F$_{1}$ & 50  & CS  & 72.7 & -0.6 & 5.0 & 4.7 (91.9) & 0.8 (24.7) \\
F$_{1}$ & 100 & AR1 & 72.7 & 0.0  & 0.7 & 0.7 (95.0) & 0.5 (83.8) \\
F$_{1}$ & 100 & CS  & 72.7 & -0.3 & 3.4 & 3.3 (94.2) & 0.5 (25.0) \\
F$_{1}$ & 200 & AR1 & 72.7 & 0.0  & 0.5 & 0.5 (96.0) & 0.4 (85.8) \\
F$_{1}$ & 200 & CS  & 72.7 & -0.1 & 2.4 & 2.4 (94.4) & 0.4 (24.9) \\
\bottomrule
\end{tabular}
\end{table}

\begin{table}[tbp]\footnotesize
\centering
\caption{Simulation results in the binary setting under moderate within-cluster dependence (\(\rho=0.5\)). ESE indicates empirical standard error. ASE$_{\mathrm{rob}}$ (CP$_{\mathrm{rob}}$) and ASE$_{\mathrm{naive}}$ (CP$_{\mathrm{naive}}$) denote cluster-robust and naive asymptotic standard errors with corresponding coverage probabilities. True performance values, bias, and all SEs are reported on the \(\times 100\) scale, and CP is reported as a percentage (ideally 95\%).}
\label{tab:sim_binary_rho05}
\begin{tabular}{lccrrrrr}
\toprule
Metric & $n$ & Corr. & True & Bias & ESE & ASE$_{\mathrm{rob}}$ (CP) & ASE$_{\mathrm{naive}}$ (CP) \\
\midrule
\multicolumn{8}{l}{\textbf{Balanced:} $P(\text{Class 1})=0.5$, $P(\text{Class 2})=0.5$}\\
\midrule
Sensitivity & 50  & AR1 & 70.0 & 0.0  & 0.8 & 0.8 (93.0) & 0.6 (89.4) \\
Sensitivity & 50  & CS  & 70.0 & -0.3 & 3.3 & 3.3 (94.1) & 0.7 (29.4) \\
Sensitivity & 100 & AR1 & 70.0 & 0.0  & 0.5 & 0.5 (95.2) & 0.5 (90.3) \\
Sensitivity & 100 & CS  & 70.0 & 0.0  & 2.4 & 2.3 (94.4) & 0.5 (28.6) \\
Sensitivity & 200 & AR1 & 70.0 & 0.0  & 0.4 & 0.4 (94.7) & 0.3 (90.6) \\
Sensitivity & 200 & CS  & 70.0 & 0.0  & 1.6 & 1.7 (95.1) & 0.3 (30.3) \\

Specificity & 50  & AR1 & 70.0 & 0.0  & 0.8 & 0.8 (94.2) & 0.6 (90.8) \\
Specificity & 50  & CS  & 70.0 & -0.2 & 3.4 & 3.3 (93.4) & 0.6 (28.8) \\
Specificity & 100 & AR1 & 70.0 & 0.0  & 0.5 & 0.5 (94.2) & 0.5 (89.1) \\
Specificity & 100 & CS  & 70.0 & -0.1 & 2.4 & 2.3 (93.7) & 0.5 (29.8) \\
Specificity & 200 & AR1 & 70.0 & 0.0  & 0.4 & 0.4 (95.2) & 0.3 (91.0) \\
Specificity & 200 & CS  & 70.0 & 0.0  & 1.6 & 1.7 (95.4) & 0.3 (30.8) \\

MCC & 50  & AR1 & 40.0 & 0.0  & 1.0 & 1.0 (93.7) & 0.9 (93.2) \\
MCC & 50  & CS  & 40.0 & -0.5 & 3.4 & 3.3 (93.2) & 0.9 (40.2) \\
MCC & 100 & AR1 & 40.0 & 0.0  & 0.7 & 0.7 (95.3) & 0.6 (93.6) \\
MCC & 100 & CS  & 40.0 & -0.1 & 2.4 & 2.4 (94.0) & 0.6 (39.3) \\
MCC & 200 & AR1 & 40.0 & 0.0  & 0.5 & 0.5 (95.0) & 0.5 (93.2) \\
MCC & 200 & CS  & 40.0 & 0.0  & 1.7 & 1.7 (94.9) & 0.5 (43.4) \\

\midrule
\multicolumn{8}{l}{\textbf{Imbalanced:} $P(\text{Class 1})=0.8$, $P(\text{Class 2})=0.2$}\\
\midrule
Sensitivity & 50  & AR1 & 80.0 & 0.0  & 0.9 & 0.9 (94.7) & 0.9 (94.6) \\
Sensitivity & 50  & CS  & 80.0 & -0.3 & 2.6 & 2.5 (91.6) & 0.9 (51.0) \\
Sensitivity & 100 & AR1 & 80.0 & 0.0  & 0.7 & 0.7 (95.0) & 0.6 (94.3) \\
Sensitivity & 100 & CS  & 80.0 & -0.1 & 1.9 & 1.8 (93.3) & 0.6 (50.2) \\
Sensitivity & 200 & AR1 & 80.0 & 0.0  & 0.5 & 0.5 (94.4) & 0.4 (93.8) \\
Sensitivity & 200 & CS  & 80.0 & -0.1 & 1.3 & 1.3 (93.6) & 0.4 (50.3) \\

Specificity & 50  & AR1 & 90.0 & 0.0  & 0.4 & 0.4 (94.2) & 0.3 (91.9) \\
Specificity & 50  & CS  & 90.0 & -0.1 & 1.3 & 1.3 (94.0) & 0.3 (37.0) \\
Specificity & 100 & AR1 & 90.0 & 0.0  & 0.3 & 0.3 (94.5) & 0.2 (91.8) \\
Specificity & 100 & CS  & 90.0 & 0.0  & 1.0 & 1.0 (94.8) & 0.2 (35.9) \\
Specificity & 200 & AR1 & 90.0 & 0.0  & 0.2 & 0.2 (95.4) & 0.2 (92.4) \\
Specificity & 200 & CS  & 90.0 & 0.0  & 0.7 & 0.7 (93.8) & 0.2 (37.1) \\

MCC & 50  & AR1 & 65.6 & 0.0  & 0.9 & 0.9 (93.6) & 0.9 (93.8) \\
MCC & 50  & CS  & 65.6 & -0.4 & 2.6 & 2.4 (90.1) & 0.9 (52.1) \\
MCC & 100 & AR1 & 65.6 & 0.0  & 0.7 & 0.7 (94.1) & 0.6 (93.6) \\
MCC & 100 & CS  & 65.6 & -0.1 & 1.8 & 1.8 (92.8) & 0.6 (52.0) \\
MCC & 200 & AR1 & 65.6 & 0.0  & 0.5 & 0.5 (94.8) & 0.5 (93.8) \\
MCC & 200 & CS  & 65.6 & -0.1 & 1.3 & 1.3 (93.5) & 0.5 (49.6) \\
\bottomrule
\end{tabular}
\end{table}

\begin{table}[tbp]\footnotesize
\centering
\caption{Simulation results in the binary setting under moderate within-cluster dependence (\(\rho=0.5\)) for accuracy and \(F_1\). ESE indicates empirical standard error. ASE$_{\mathrm{rob}}$ (CP$_{\mathrm{rob}}$) and ASE$_{\mathrm{naive}}$ (CP$_{\mathrm{naive}}$) denote cluster-robust and naive asymptotic standard errors with corresponding coverage probabilities. True performance values, bias, and all SEs are reported on the \(\times 100\) scale, and CP is reported as a percentage (ideally 95\%).}
\label{tab:sim_binary_rho05_accf1}
\begin{tabular}{lccrrrrr}
\toprule
Metric & $n$ & Corr. & True & Bias & ESE & ASE$_{\mathrm{rob}}$ (CP) & ASE$_{\mathrm{naive}}$ (CP) \\
\midrule
\multicolumn{8}{l}{\textbf{Balanced:} $P(\text{Class 1})=0.5$, $P(\text{Class 2})=0.5$}\\
\midrule
Accuracy & 50  & AR1 & 70.0 & 0.0  & 0.5 & 0.5 (93.7) & 0.5 (93.3) \\
Accuracy & 50  & CS  & 70.0 & -0.1 & 1.7 & 1.7 (93.8) & 0.5 (40.1) \\
Accuracy & 100 & AR1 & 70.0 & 0.0  & 0.3 & 0.3 (95.2) & 0.3 (93.7) \\
Accuracy & 100 & CS  & 70.0 & 0.0  & 1.2 & 1.2 (94.4) & 0.3 (39.6) \\
Accuracy & 200 & AR1 & 70.0 & 0.0  & 0.2 & 0.2 (95.0) & 0.2 (93.2) \\
Accuracy & 200 & CS  & 70.0 & 0.0  & 0.8 & 0.8 (95.1) & 0.2 (42.9) \\

F$_{1}$ & 50  & AR1 & 70.0 & 0.0  & 0.6 & 0.6 (93.3) & 0.5 (88.2) \\
F$_{1}$ & 50  & CS  & 70.0 & -0.3 & 2.9 & 2.9 (94.4) & 0.5 (27.0) \\
F$_{1}$ & 100 & AR1 & 70.0 & 0.0  & 0.4 & 0.4 (94.2) & 0.4 (89.7) \\
F$_{1}$ & 100 & CS  & 70.0 & 0.0  & 2.1 & 2.0 (94.2) & 0.4 (27.0) \\
F$_{1}$ & 200 & AR1 & 70.0 & 0.0  & 0.3 & 0.3 (94.7) & 0.3 (89.4) \\
F$_{1}$ & 200 & CS  & 70.0 & 0.0  & 1.4 & 1.4 (95.1) & 0.3 (27.4) \\

\midrule
\multicolumn{8}{l}{\textbf{Imbalanced:} $P(\text{Class 1})=0.8$, $P(\text{Class 2})=0.2$}\\
\midrule
Accuracy & 50  & AR1 & 88.0 & 0.0  & 0.4 & 0.4 (94.6) & 0.3 (93.1) \\
Accuracy & 50  & CS  & 88.0 & 0.0  & 1.2 & 1.2 (93.8) & 0.3 (41.7) \\
Accuracy & 100 & AR1 & 88.0 & 0.0  & 0.3 & 0.2 (94.8) & 0.2 (92.2) \\
Accuracy & 100 & CS  & 88.0 & 0.0  & 0.8 & 0.8 (95.3) & 0.2 (40.3) \\
Accuracy & 200 & AR1 & 88.0 & 0.0  & 0.2 & 0.2 (95.1) & 0.2 (92.8) \\
Accuracy & 200 & CS  & 88.0 & 0.0  & 0.6 & 0.6 (94.4) & 0.2 (41.8) \\

F$_{1}$ & 50  & AR1 & 72.7 & 0.0  & 0.8 & 0.8 (93.6) & 0.8 (92.4) \\
F$_{1}$ & 50  & CS  & 72.7 & -0.4 & 2.8 & 2.6 (91.7) & 0.8 (41.2) \\
F$_{1}$ & 100 & AR1 & 72.7 & 0.0  & 0.6 & 0.6 (94.1) & 0.5 (92.7) \\
F$_{1}$ & 100 & CS  & 72.7 & -0.1 & 2.0 & 1.9 (93.2) & 0.5 (40.5) \\
F$_{1}$ & 200 & AR1 & 72.7 & 0.0  & 0.4 & 0.4 (94.7) & 0.4 (92.9) \\
F$_{1}$ & 200 & CS  & 72.7 & -0.1 & 1.4 & 1.4 (93.7) & 0.4 (39.6) \\
\bottomrule
\end{tabular}
\end{table}

\begin{table}[tbp]\footnotesize
\centering
\caption{Simulation results in the multiclass setting under moderate within-cluster dependence (\(\rho=0.5\)). ESE indicates empirical standard error. ASE$_{\mathrm{rob}}$ (CP$_{\mathrm{rob}}$) and ASE$_{\mathrm{naive}}$ (CP$_{\mathrm{naive}}$) denote cluster-robust and naive asymptotic standard errors with corresponding coverage probabilities. True performance values, bias, and all SEs are reported on the \(\times 100\) scale, and CP is reported as a percentage (ideally 95\%).}
\label{tab:sim_multiclass_rho05}
\begin{tabular}{lccrrrrr}
\toprule
Metric & $n$ & Corr. & True & Bias & ESE & ASE$_{\mathrm{rob}}$ (CP) & ASE$_{\mathrm{naive}}$ (CP) \\
\midrule
\multicolumn{8}{l}{\textbf{Balanced:} $P(\text{Class 1})=0.3$, $P(\text{Class 2})=0.37$, $P(\text{Class 3})=0.33$}\\
\midrule
macro-F$_{1}$ & 50  & AR1 & 82.0 & 0.0  & 0.4 & 0.4 (94.0) & 0.4 (94.6) \\
macro-F$_{1}$ & 50  & CS  & 82.0 & -0.2 & 0.9 & 0.8 (92.5) & 0.4 (60.4) \\
macro-F$_{1}$ & 100 & AR1 & 82.0 & 0.0  & 0.3 & 0.3 (95.3) & 0.3 (95.4) \\
macro-F$_{1}$ & 100 & CS  & 82.0 & 0.0  & 0.6 & 0.6 (94.1) & 0.3 (62.4) \\
macro-F$_{1}$ & 200 & AR1 & 82.0 & 0.0  & 0.2 & 0.2 (95.3) & 0.2 (95.0) \\
macro-F$_{1}$ & 200 & CS  & 82.0 & 0.0  & 0.4 & 0.4 (94.8) & 0.2 (64.8) \\

micro-F$_{1}$ & 50  & AR1 & 82.0 & 0.0  & 0.4 & 0.4 (93.8) & 0.4 (94.4) \\
micro-F$_{1}$ & 50  & CS  & 82.0 & 0.0  & 0.8 & 0.8 (93.3) & 0.4 (63.7) \\
micro-F$_{1}$ & 100 & AR1 & 82.0 & 0.0  & 0.3 & 0.3 (95.3) & 0.3 (95.4) \\
micro-F$_{1}$ & 100 & CS  & 82.0 & 0.0  & 0.6 & 0.6 (94.8) & 0.3 (63.8) \\
micro-F$_{1}$ & 200 & AR1 & 82.0 & 0.0  & 0.2 & 0.2 (95.3) & 0.2 (95.0) \\
micro-F$_{1}$ & 200 & CS  & 82.0 & 0.0  & 0.4 & 0.4 (94.7) & 0.2 (66.9) \\

\midrule
\multicolumn{8}{l}{\textbf{Imbalanced:} $P(\text{Class 1})=0.47$, $P(\text{Class 2})=0.26$, $P(\text{Class 3})=0.27$}\\
\midrule
macro-F$_{1}$ & 50  & AR1 & 75.6 & 0.0  & 0.5 & 0.5 (93.8) & 0.5 (94.6) \\
macro-F$_{1}$ & 50  & CS  & 75.6 & -0.2 & 1.0 & 1.0 (91.6) & 0.5 (59.8) \\
macro-F$_{1}$ & 100 & AR1 & 75.6 & 0.0  & 0.3 & 0.3 (95.0) & 0.3 (94.5) \\
macro-F$_{1}$ & 100 & CS  & 75.6 & -0.1 & 0.7 & 0.7 (94.3) & 0.3 (60.2) \\
macro-F$_{1}$ & 200 & AR1 & 75.6 & 0.0  & 0.2 & 0.2 (94.7) & 0.2 (94.4) \\
macro-F$_{1}$ & 200 & CS  & 75.6 & 0.0  & 0.5 & 0.5 (94.2) & 0.2 (61.3) \\

micro-F$_{1}$ & 50  & AR1 & 78.0 & 0.0  & 0.4 & 0.4 (94.0) & 0.4 (93.3) \\
micro-F$_{1}$ & 50  & CS  & 78.0 & 0.0  & 1.3 & 1.3 (94.2) & 0.4 (47.6) \\
micro-F$_{1}$ & 100 & AR1 & 78.0 & 0.0  & 0.3 & 0.3 (94.6) & 0.3 (93.7) \\
micro-F$_{1}$ & 100 & CS  & 78.0 & 0.0  & 0.9 & 0.9 (94.6) & 0.3 (46.2) \\
micro-F$_{1}$ & 200 & AR1 & 78.0 & 0.0  & 0.2 & 0.2 (94.8) & 0.2 (93.5) \\
micro-F$_{1}$ & 200 & CS  & 78.0 & 0.0  & 0.7 & 0.7 (93.9) & 0.2 (46.7) \\
\bottomrule
\end{tabular}
\end{table}

\begin{table}[tbp]\footnotesize
\centering
\caption{Supplementary simulation results in the balanced binary setting with a very limited number of clusters (\(n=25\)). ESE indicates empirical standard error. ASE$_{\mathrm{rob}}$ (CP$_{\mathrm{rob}}$) and ASE$_{\mathrm{naive}}$ (CP$_{\mathrm{naive}}$) denote cluster-robust and naive asymptotic standard errors with corresponding coverage probabilities. True performance values, bias, and all SEs are reported on the \(\times 100\) scale, and CP is reported as a percentage.}
\label{tab:sim_binary_balanced_n25}
\begin{tabular}{lccrrrrr}
\toprule
Metric & Scenario & Corr. & True & Bias & ESE & ASE$_{\mathrm{rob}}$ (CP) & ASE$_{\mathrm{naive}}$ (CP) \\
\midrule
\multicolumn{8}{l}{\textbf{Balanced:} \(P(\text{Class 1})=0.5\), \(P(\text{Class 2})=0.5\)}\\
\midrule
Accuracy    & \(\rho=0.5\) & AR1 & 70.0 & 0.0  & 0.7 & 0.7 (92.9) & 0.6 (92.8) \\
Accuracy    & \(\rho=0.8\) & AR1 & 70.0 & 0.0  & 0.9 & 0.9 (93.5) & 0.6 (84.0) \\
Accuracy    & \(\rho=0.5\) & CS  & 70.0 & 0.0  & 2.4 & 2.3 (93.0) & 0.6 (38.6) \\
Accuracy    & \(\rho=0.8\) & CS  & 70.0 & -0.1 & 4.7 & 4.5 (92.6) & 0.6 (21.7) \\

F$_1$      & \(\rho=0.5\) & AR1 & 70.0 & 0.0  & 0.9 & 0.9 (93.0) & 0.7 (88.7) \\
F$_1$      & \(\rho=0.8\) & AR1 & 70.0 & -0.1 & 1.3 & 1.3 (92.9) & 0.7 (72.0) \\
F$_1$      & \(\rho=0.5\) & CS  & 70.0 & -0.2 & 4.1 & 4.0 (92.1) & 0.7 (28.5) \\
F$_1$      & \(\rho=0.8\) & CS  & 70.0 & -0.8 & 6.5 & 6.2 (92.5) & 0.8 (17.6) \\

MCC         & \(\rho=0.5\) & AR1 & 40.0 & 0.0  & 1.4 & 1.3 (92.8) & 1.3 (92.9) \\
MCC         & \(\rho=0.8\) & AR1 & 40.0 & 0.0  & 1.8 & 1.8 (93.4) & 1.3 (83.7) \\
MCC         & \(\rho=0.5\) & CS  & 40.0 & -0.7 & 4.8 & 4.6 (91.5) & 1.3 (38.1) \\
MCC         & \(\rho=0.8\) & CS  & 40.0 & -1.3 & 9.4 & 9.0 (91.8) & 1.3 (19.9) \\

Precision   & \(\rho=0.5\) & AR1 & 70.0 & 0.0  & 1.0 & 1.0 (93.7) & 0.9 (92.6) \\
Precision   & \(\rho=0.8\) & AR1 & 70.0 & 0.0  & 1.4 & 1.3 (93.7) & 0.9 (80.5) \\
Precision   & \(\rho=0.5\) & CS  & 70.0 & -0.2 & 3.7 & 3.5 (91.9) & 0.9 (37.1) \\
Precision   & \(\rho=0.8\) & CS  & 70.0 & -0.6 & 6.0 & 5.7 (91.7) & 0.9 (23.6) \\

Sensitivity & \(\rho=0.5\) & AR1 & 70.0 & 0.0  & 1.1 & 1.1 (92.6) & 0.9 (89.0) \\
Sensitivity & \(\rho=0.8\) & AR1 & 70.0 & -0.1 & 1.6 & 1.5 (92.8) & 0.9 (73.6) \\
Sensitivity & \(\rho=0.5\) & CS  & 70.0 & -0.2 & 4.7 & 4.5 (92.0) & 0.9 (31.4) \\
Sensitivity & \(\rho=0.8\) & CS  & 70.0 & -0.9 & 7.3 & 7.1 (92.8) & 0.9 (20.0) \\

Specificity & \(\rho=0.5\) & AR1 & 70.0 & 0.0  & 1.1 & 1.1 (93.2) & 0.9 (90.9) \\
Specificity & \(\rho=0.8\) & AR1 & 70.0 & 0.1  & 1.6 & 1.5 (93.6) & 0.9 (74.2) \\
Specificity & \(\rho=0.5\) & CS  & 70.0 & -0.6 & 4.8 & 4.6 (92.9) & 0.9 (27.8) \\
Specificity & \(\rho=0.8\) & CS  & 70.0 & -0.6 & 7.3 & 7.0 (92.8) & 0.9 (19.4) \\
\bottomrule
\end{tabular}
\end{table}

\begin{table}[tbp]\footnotesize
\centering
\caption{Supplementary simulation results in the imbalanced binary setting with a very limited number of clusters (\(n=25\)). ESE indicates empirical standard error. ASE$_{\mathrm{rob}}$ (CP$_{\mathrm{rob}}$) and ASE$_{\mathrm{naive}}$ (CP$_{\mathrm{naive}}$) denote cluster-robust and naive asymptotic standard errors with corresponding coverage probabilities. True performance values, bias, and all SEs are reported on the \(\times 100\) scale, and CP is reported as a percentage.}
\label{tab:sim_binary_imbalanced_n25}
\begin{tabular}{lccrrrrr}
\toprule
Metric & Scenario & Corr. & True & Bias & ESE & ASE$_{\mathrm{rob}}$ (CP) & ASE$_{\mathrm{naive}}$ (CP) \\
\midrule
\multicolumn{8}{l}{\textbf{Imbalanced:} \(P(\text{Class 1})=0.8\), \(P(\text{Class 2})=0.2\)}\\
\midrule
Accuracy    & \(\rho=0.5\) & AR1 & 88.0 & 0.0  & 0.5 & 0.5 (93.2) & 0.5 (92.5) \\
Accuracy    & \(\rho=0.8\) & AR1 & 88.0 & 0.0  & 0.6 & 0.6 (93.6) & 0.5 (84.2) \\
Accuracy    & \(\rho=0.5\) & CS  & 88.0 & 0.0  & 1.6 & 1.6 (93.2) & 0.5 (41.1) \\
Accuracy    & \(\rho=0.8\) & CS  & 88.0 & 0.0  & 2.7 & 2.6 (91.6) & 0.5 (27.4) \\

F$_1$      & \(\rho=0.5\) & AR1 & 72.7 & 0.0  & 1.2 & 1.1 (93.2) & 1.1 (92.4) \\
F$_1$      & \(\rho=0.8\) & AR1 & 72.7 & 0.0  & 1.5 & 1.4 (93.2) & 1.1 (84.6) \\
F$_1$      & \(\rho=0.5\) & CS  & 72.7 & -0.7 & 4.3 & 3.5 (86.2) & 1.1 (38.6) \\
F$_1$      & \(\rho=0.8\) & CS  & 72.7 & -1.5 & 7.5 & 6.3 (85.7) & 1.1 (24.3) \\

MCC         & \(\rho=0.5\) & AR1 & 65.6 & 0.1  & 1.4 & 1.3 (93.2) & 1.3 (93.0) \\
MCC         & \(\rho=0.8\) & AR1 & 65.6 & -0.1 & 1.6 & 1.5 (92.8) & 1.3 (88.6) \\
MCC         & \(\rho=0.5\) & CS  & 65.6 & -0.8 & 3.9 & 3.2 (85.2) & 1.3 (49.3) \\
MCC         & \(\rho=0.8\) & CS  & 65.6 & -1.6 & 7.6 & 6.5 (86.7) & 1.3 (26.5) \\

Precision   & \(\rho=0.5\) & AR1 & 66.7 & 0.0  & 1.5 & 1.4 (92.6) & 1.4 (93.0) \\
Precision   & \(\rho=0.8\) & AR1 & 66.7 & 0.0  & 1.7 & 1.7 (93.0) & 1.4 (87.7) \\
Precision   & \(\rho=0.5\) & CS  & 66.7 & -0.7 & 4.5 & 3.7 (85.4) & 1.4 (44.6) \\
Precision   & \(\rho=0.8\) & CS  & 66.7 & -1.4 & 7.8 & 6.6 (84.8) & 1.4 (28.1) \\

Sensitivity & \(\rho=0.5\) & AR1 & 80.0 & 0.0  & 1.4 & 1.3 (92.6) & 1.3 (93.2) \\
Sensitivity & \(\rho=0.8\) & AR1 & 80.0 & -0.1 & 1.6 & 1.5 (93.6) & 1.3 (88.6) \\
Sensitivity & \(\rho=0.5\) & CS  & 80.0 & -0.7 & 4.0 & 3.3 (86.4) & 1.3 (49.5) \\
Sensitivity & \(\rho=0.8\) & CS  & 80.0 & -1.6 & 7.0 & 5.9 (87.4) & 1.3 (33.1) \\

Specificity & \(\rho=0.5\) & AR1 & 90.0 & 0.0  & 0.5 & 0.5 (93.0) & 0.5 (91.2) \\
Specificity & \(\rho=0.8\) & AR1 & 90.0 & 0.0  & 0.7 & 0.7 (92.6) & 0.5 (82.1) \\
Specificity & \(\rho=0.5\) & CS  & 90.0 & -0.1 & 2.0 & 1.9 (93.3) & 0.5 (36.1) \\
Specificity & \(\rho=0.8\) & CS  & 90.0 & -0.1 & 3.0 & 2.8 (91.8) & 0.5 (26.3) \\
\bottomrule
\end{tabular}
\end{table}

\begin{table}[tbp]\footnotesize
\centering
\caption{Supplementary simulation results in the multiclass setting with a very limited number of clusters (\(n=25\)). ESE indicates empirical standard error. ASE$_{\mathrm{rob}}$ (CP$_{\mathrm{rob}}$) and ASE$_{\mathrm{naive}}$ (CP$_{\mathrm{naive}}$) denote cluster-robust and naive asymptotic standard errors with corresponding coverage probabilities. True performance values, bias, and all SEs are reported on the \(\times 100\) scale, and CP is reported as a percentage.}
\label{tab:sim_multiclass_n25}
\begin{tabular}{lccrrrrr}
\toprule
Metric & Scenario & Corr. & True & Bias & ESE & ASE$_{\mathrm{rob}}$ (CP) & ASE$_{\mathrm{naive}}$ (CP) \\
\midrule
\multicolumn{8}{l}{\textbf{Balanced:} $P(\text{Class 1})=0.3$, $P(\text{Class 2})=0.37$, $P(\text{Class 3})=0.33$}\\
\midrule
macro-F$_1$ & \(\rho=0.5\) & AR1 & 82.0 & 0.0  & 0.6 & 0.5 (93.0) & 0.5 (93.8) \\
macro-F$_1$ & \(\rho=0.8\) & AR1 & 82.0 & 0.0  & 0.6 & 0.6 (93.0) & 0.5 (91.0) \\
macro-F$_1$ & \(\rho=0.5\) & CS  & 82.0 & -0.3 & 1.2 & 1.1 (90.4) & 0.6 (62.2) \\
macro-F$_1$ & \(\rho=0.8\) & CS  & 82.0 & -0.6 & 2.0 & 1.9 (91.8) & 0.6 (41.1) \\

micro-F$_1$ & \(\rho=0.5\) & AR1 & 82.0 & 0.0  & 0.6 & 0.5 (93.0) & 0.5 (93.9) \\
micro-F$_1$ & \(\rho=0.8\) & AR1 & 82.0 & 0.0  & 0.6 & 0.6 (92.8) & 0.5 (91.7) \\
micro-F$_1$ & \(\rho=0.5\) & CS  & 82.0 & 0.0  & 1.1 & 1.1 (92.4) & 0.5 (65.0) \\
micro-F$_1$ & \(\rho=0.8\) & CS  & 82.0 & 0.0  & 1.8 & 1.8 (92.6) & 0.5 (43.8) \\

\midrule
\multicolumn{8}{l}{\textbf{Imbalanced:} $P(\text{Class 1})=0.47$, $P(\text{Class 2})=0.26$, $P(\text{Class 3})=0.27$}\\
\midrule
macro-F$_1$ & \(\rho=0.5\) & AR1 & 75.6 & 0.0  & 0.7 & 0.6 (92.7) & 0.6 (94.5) \\
macro-F$_1$ & \(\rho=0.8\) & AR1 & 75.6 & 0.0  & 0.7 & 0.7 (93.6) & 0.6 (91.4) \\
macro-F$_1$ & \(\rho=0.5\) & CS  & 75.6 & -0.3 & 1.5 & 1.3 (88.1) & 0.6 (60.0) \\
macro-F$_1$ & \(\rho=0.8\) & CS  & 75.6 & -0.7 & 2.5 & 2.3 (89.9) & 0.7 (38.8) \\

micro-F$_1$ & \(\rho=0.5\) & AR1 & 78.0 & 0.0  & 0.6 & 0.6 (92.8) & 0.6 (93.3) \\
micro-F$_1$ & \(\rho=0.8\) & AR1 & 78.0 & 0.0  & 0.8 & 0.7 (93.7) & 0.6 (86.7) \\
micro-F$_1$ & \(\rho=0.5\) & CS  & 78.0 & 0.0  & 1.9 & 1.8 (91.3) & 0.6 (44.9) \\
micro-F$_1$ & \(\rho=0.8\) & CS  & 78.0 & -0.1 & 3.0 & 2.9 (92.9) & 0.6 (29.1) \\
\bottomrule
\end{tabular}
\end{table}

\begin{table}[tbp]\footnotesize
\centering
\caption{Supplementary simulation results in the rare-event binary setting with positive-class prevalence equal to \(1\%\), with sensitivity and specificity both set to \(0.7\), for \(F_1\) and MCC. ESE indicates empirical standard error. ASE$_{\mathrm{rob}}$ (CP$_{\mathrm{rob}}$) and ASE$_{\mathrm{naive}}$ (CP$_{\mathrm{naive}}$) denote cluster-robust and naive asymptotic standard errors with corresponding coverage probabilities. True performance values, bias, and all SEs are reported on the \(\times 100\) scale, and CP is reported as a percentage.}
\label{tab:sim_binary_rare_moderate_f1mcc}
\begin{tabular}{lcccrrrrr}
\toprule
Metric & \(n\) & Corr. & \(\rho\) & True & Bias & ESE & ASE$_{\mathrm{rob}}$ (CP) & ASE$_{\mathrm{naive}}$ (CP) \\
\midrule
F$_1$ & 25  & AR1 & 0.5 & 4.5 & 0.0  & 0.8 & 0.8 (92.1) & 0.7 (91.3) \\
F$_1$ & 25  & AR1 & 0.8 & 4.5 & 0.0  & 1.2 & 1.1 (88.7) & 0.7 (76.1) \\
F$_1$ & 25  & CS  & 0.5 & 4.5 & -0.2 & 3.1 & 2.0 (67.3) & 0.7 (33.4) \\
F$_1$ & 25  & CS  & 0.8 & 4.5 & -0.4 & 5.7 & 2.8 (46.9) & 0.5 (13.2) \\
F$_1$ & 50  & AR1 & 0.5 & 4.5 & 0.0  & 0.6 & 0.6 (94.3) & 0.5 (93.3) \\
F$_1$ & 50  & AR1 & 0.8 & 4.5 & 0.0  & 0.8 & 0.8 (92.8) & 0.5 (80.0) \\
F$_1$ & 50  & CS  & 0.5 & 4.5 & -0.2 & 2.2 & 1.7 (74.8) & 0.5 (32.6) \\
F$_1$ & 50  & CS  & 0.8 & 4.5 & 0.0  & 4.3 & 2.8 (61.9) & 0.4 (15.1) \\
F$_1$ & 100 & AR1 & 0.5 & 4.5 & 0.0  & 0.4 & 0.4 (93.7) & 0.4 (91.6) \\
F$_1$ & 100 & AR1 & 0.8 & 4.5 & 0.0  & 0.6 & 0.6 (93.0) & 0.4 (79.2) \\
F$_1$ & 100 & CS  & 0.5 & 4.5 & 0.1  & 1.6 & 1.4 (83.6) & 0.4 (32.9) \\
F$_1$ & 100 & CS  & 0.8 & 4.5 & -0.2 & 3.1 & 2.3 (70.0) & 0.3 (15.2) \\
F$_1$ & 200 & AR1 & 0.5 & 4.5 & 0.0  & 0.3 & 0.3 (94.6) & 0.3 (92.1) \\
F$_1$ & 200 & AR1 & 0.8 & 4.5 & 0.0  & 0.4 & 0.4 (93.7) & 0.3 (78.5) \\
F$_1$ & 200 & CS  & 0.5 & 4.5 & 0.0  & 1.1 & 1.0 (87.4) & 0.3 (35.0) \\
F$_1$ & 200 & CS  & 0.8 & 4.5 & 0.0  & 2.2 & 1.9 (81.7) & 0.2 (16.0) \\
\midrule
MCC & 25  & AR1 & 0.5 & 8.7 & 0.0  & 1.6 & 1.6 (92.8) & 1.5 (92.7) \\
MCC & 25  & AR1 & 0.8 & 8.7 & -0.2 & 2.2 & 2.0 (90.2) & 1.5 (80.8) \\
MCC & 25  & CS  & 0.5 & 8.7 & -1.2 & 5.0 & 3.3 (70.2) & 1.5 (40.8) \\
MCC & 25  & CS  & 0.8 & 8.7 & -3.3 & 8.7 & 4.2 (50.8) & 1.3 (15.9) \\
MCC & 50  & AR1 & 0.5 & 8.7 & 0.0  & 1.1 & 1.1 (94.5) & 1.1 (93.2) \\
MCC & 50  & AR1 & 0.8 & 8.7 & 0.0  & 1.5 & 1.5 (93.4) & 1.1 (83.4) \\
MCC & 50  & CS  & 0.5 & 8.7 & -0.8 & 3.8 & 2.9 (78.3) & 1.1 (38.8) \\
MCC & 50  & CS  & 0.8 & 8.7 & -1.7 & 7.0 & 4.5 (66.0) & 1.0 (18.3) \\
MCC & 100 & AR1 & 0.5 & 8.7 & 0.0  & 0.8 & 0.8 (94.3) & 0.8 (93.0) \\
MCC & 100 & AR1 & 0.8 & 8.7 & -0.1 & 1.1 & 1.1 (94.3) & 0.8 (84.0) \\
MCC & 100 & CS  & 0.5 & 8.7 & -0.2 & 2.8 & 2.4 (85.8) & 0.8 (39.1) \\
MCC & 100 & CS  & 0.8 & 8.7 & -1.2 & 5.4 & 3.9 (72.8) & 0.7 (18.8) \\
MCC & 200 & AR1 & 0.5 & 8.7 & 0.0  & 0.6 & 0.6 (94.6) & 0.5 (93.2) \\
MCC & 200 & AR1 & 0.8 & 8.7 & 0.0  & 0.8 & 0.8 (94.4) & 0.5 (83.0) \\
MCC & 200 & CS  & 0.5 & 8.7 & -0.2 & 2.0 & 1.8 (89.7) & 0.5 (40.0) \\
MCC & 200 & CS  & 0.8 & 8.7 & -0.6 & 3.9 & 3.3 (84.4) & 0.5 (19.6) \\
\bottomrule
\end{tabular}
\end{table}

\begin{table}[tbp]\footnotesize
\centering
\caption{Supplementary simulation results in the rare-event binary setting, with sensitivity set to \(0.80\) and specificity set to \(0.90\), strong performance setting, for \(F_1\) and MCC. ESE indicates empirical standard error. ASE$_{\mathrm{rob}}$ (CP$_{\mathrm{rob}}$) and ASE$_{\mathrm{naive}}$ (CP$_{\mathrm{naive}}$) denote cluster-robust and naive asymptotic standard errors with corresponding coverage probabilities. True performance values, bias, and all SEs are reported on the \(\times 100\) scale, and CP is reported as a percentage.}
\label{tab:sim_binary_rare_strong_f1mcc}
\begin{tabular}{lcccrrrrr}
\toprule
Metric & \(n\) & Corr. & \(\rho\) & True & Bias & ESE & ASE$_{\mathrm{rob}}$ (CP) & ASE$_{\mathrm{naive}}$ (CP) \\
\midrule
F$_1$ & 25  & AR1 & 0.5 & 13.7 & 0.0  & 2.1 & 2.0 (92.4) & 1.9 (92.3) \\
F$_1$ & 25  & AR1 & 0.8 & 13.7 & 0.0  & 2.8 & 2.7 (91.6) & 1.9 (81.7) \\
F$_1$ & 25  & CS  & 0.5 & 13.7 & -1.0 & 6.5 & 3.9 (66.1) & 1.8 (40.8) \\
F$_1$ & 25  & CS  & 0.8 & 13.7 & -2.8 & 12.5 & 5.0 (41.1) & 1.4 (14.1) \\
F$_1$ & 50  & AR1 & 0.5 & 13.7 & 0.0  & 1.5 & 1.5 (93.2) & 1.4 (92.2) \\
F$_1$ & 50  & AR1 & 0.8 & 13.7 & -0.1 & 2.1 & 1.9 (91.6) & 1.4 (79.0) \\
F$_1$ & 50  & CS  & 0.5 & 13.7 & -0.7 & 4.9 & 3.5 (74.0) & 1.3 (37.0) \\
F$_1$ & 50  & CS  & 0.8 & 13.7 & -1.4 & 10.0 & 5.8 (58.9) & 1.2 (15.8) \\
F$_1$ & 100 & AR1 & 0.5 & 13.7 & -0.1 & 1.0 & 1.0 (95.3) & 1.0 (93.7) \\
F$_1$ & 100 & AR1 & 0.8 & 13.7 & 0.0  & 1.4 & 1.4 (94.0) & 1.0 (81.0) \\
F$_1$ & 100 & CS  & 0.5 & 13.7 & -0.4 & 3.6 & 2.9 (81.0) & 1.0 (38.9) \\
F$_1$ & 100 & CS  & 0.8 & 13.7 & -0.8 & 7.5 & 5.5 (69.8) & 0.9 (17.0) \\
F$_1$ & 200 & AR1 & 0.5 & 13.7 & 0.0  & 0.7 & 0.7 (95.0) & 0.7 (93.0) \\
F$_1$ & 200 & AR1 & 0.8 & 13.7 & 0.0  & 1.0 & 1.0 (94.8) & 0.7 (81.3) \\
F$_1$ & 200 & CS  & 0.5 & 13.7 & -0.2 & 2.7 & 2.3 (85.3) & 0.7 (37.8) \\
F$_1$ & 200 & CS  & 0.8 & 13.7 & -0.3 & 5.4 & 4.6 (80.2) & 0.7 (18.5) \\
\midrule
MCC & 25  & AR1 & 0.5 & 22.5 & -0.1 & 2.6 & 2.4 (92.6) & 2.4 (92.7) \\
MCC & 25  & AR1 & 0.8 & 22.5 & -0.1 & 3.3 & 3.2 (93.1) & 2.4 (84.2) \\
MCC & 25  & CS  & 0.5 & 22.5 & -2.0 & 7.4 & 4.4 (68.3) & 2.4 (44.8) \\
MCC & 25  & CS  & 0.8 & 22.5 & -6.7 & 14.5 & 5.7 (43.2) & 2.1 (14.8) \\
MCC & 50  & AR1 & 0.5 & 22.5 & 0.0  & 1.8 & 1.8 (94.1) & 1.7 (93.2) \\
MCC & 50  & AR1 & 0.8 & 22.5 & -0.2 & 2.4 & 2.3 (93.1) & 1.7 (81.8) \\
MCC & 50  & CS  & 0.5 & 22.5 & -1.3 & 5.6 & 4.0 (75.6) & 1.7 (40.6) \\
MCC & 50  & CS  & 0.8 & 22.5 & -3.7 & 11.8 & 6.6 (62.3) & 1.6 (16.9) \\
MCC & 100 & AR1 & 0.5 & 22.5 & -0.1 & 1.2 & 1.3 (95.8) & 1.2 (94.7) \\
MCC & 100 & AR1 & 0.8 & 22.5 & -0.1 & 1.7 & 1.7 (94.0) & 1.2 (83.2) \\
MCC & 100 & CS  & 0.5 & 22.5 & -0.7 & 4.1 & 3.3 (82.4) & 1.2 (41.9) \\
MCC & 100 & CS  & 0.8 & 22.5 & -2.1 & 8.9 & 6.3 (72.6) & 1.2 (18.4) \\
MCC & 200 & AR1 & 0.5 & 22.5 & 0.0  & 0.9 & 0.9 (94.9) & 0.8 (93.8) \\
MCC & 200 & AR1 & 0.8 & 22.5 & 0.0  & 1.2 & 1.2 (94.8) & 0.8 (83.2) \\
MCC & 200 & CS  & 0.5 & 22.5 & -0.4 & 3.0 & 2.6 (86.4) & 0.8 (40.1) \\
MCC & 200 & CS  & 0.8 & 22.5 & -1.0 & 6.3 & 5.3 (81.8) & 0.8 (19.8) \\
\bottomrule
\end{tabular}
\end{table}
\newpage
\bibliographystyle{elsarticle-harv} 
\bibliography{bibdat}


\end{document}